\documentclass[
    reprint,
    amsmath, amssymb,
    aps,
    floatfix,
    nofootinbib
]{revtex4-1}

\usepackage{graphicx}
\usepackage{dcolumn}
\usepackage{bm}
\usepackage{enumerate}
\usepackage{tcolorbox}
\usepackage{tikz}
\usepackage{quantikz}
\usetikzlibrary{positioning,mindmap,trees} 
\tcbuselibrary{theorems}
\usepackage{threeparttable}
\usepackage{natbib}
\usepackage{multirow}
\usepackage{dsfont}
\usepackage{color}
\usepackage{hyperref}
\usepackage[table]{xcolor}
\usepackage[utf8]{inputenc}
\usepackage{hyperref}

\usepackage[margin=1in]{geometry}
\usepackage{amsmath,amsthm,amssymb,booktabs}
\usepackage{graphicx}
\usepackage{hyperref}
\usepackage{natbib}
\usepackage{bm}
\usepackage{mathrsfs}
\usepackage{braket}

\usepackage{tabularx}  
\usepackage{booktabs}  

\newtheorem{theorem}{Theorem}
\newtheorem{lemma}[theorem]{Lemma}

\newtheorem{definition}{Definition}
\newtheorem{proposition}{Proposition}
\newtheorem{assumption}{Assumption}
  
\usepackage{ragged2e}
\definecolor{myViolet}{RGB}{255, 154, 249}
\definecolor{myRose}{RGB}{255,169,249}

\newtcolorbox[auto counter, number within=section]{mybox}[2][]{
    colback=gray!10,
    colframe=black,
    fonttitle=\bfsection,
    title=Box~\thetcbcounter: #2,
    #1
}

\usepackage{lineno}
\usepackage{etoolbox}
\makeatletter
\pretocmd{\@outputpage@tail}{\if@twocolumn\switchlinenumbers\fi}{}{}
\usepackage{bm}
\usepackage{tikz}
\usetikzlibrary{arrows.meta, positioning}
\usepackage{mathrsfs}
\usepackage{braket}
\usepackage{tikz-cd}
\newtheorem{remark}{Remark}
\usepackage{graphicx}
\usepackage{dcolumn}
\usepackage{bm}
\usepackage{enumerate}
\usepackage{tcolorbox}
\usepackage{tikz}
\usepackage{quantikz}
\usetikzlibrary{positioning,mindmap,trees} 
\tcbuselibrary{theorems}
\usepackage{threeparttable}
\usepackage{natbib}
\usepackage{multirow}
\usepackage{dsfont}
\usepackage{color}
\usepackage{hyperref}
\usepackage[table]{xcolor
}

\begin{document}

\title{How Quantum Circuits Actually Learn: A Causal Identification of Genuine Quantum Contributions}

\author{Cyrille Yetuyetu Kesiku$^{1}$}
\email{kesiku.cyrille@deusto.es}

\author{Begonya Garcia-Zapirain$^{1}$}
\affiliation{$^{1}$Department of Computer Science, Electronics and Telecommunications, University of Deusto, Bilbao, Spain}

\date{\today}

\begin{abstract}
Attributing performance gains in quantum machine learning to genuine quantum resources  rather than to classical architectural scaling  remains an open methodological challenge. We address this by introducing a counterfactual causal mediation framework that decomposes inter-architectural performance differences into direct effects, attributable to circuit parameterization and expressivity, and indirect effects mediated by quantum information-theoretic observables: entanglement entropy, purity, linear entropy, and quantum mutual information. Applying this framework to five circuit topologies and three benchmark datasets (across 43 validated configurations) reveals that direct architectural contributions systematically exceed quantum-mediated effects, with a mean ratio of 13.1:1 and a mean indirect contribution of 0.82\%. These results suggest that current variational quantum circuits operate substantially below their quantum potential, and that principled resource-aware circuit design represents a tractable path toward measurable quantum-mediated performance gains.
\end{abstract}

\flushbottom
\maketitle

\thispagestyle{empty}

\section{INTRODUCTION}

Quantum machine learning (QML) \cite{biamonte2017quantum} has emerged as a promising avenue for achieving practical advantages in the Noisy Intermediate-Scale Quantum (NISQ) era \cite{preskill2018quantum}. The central expectation is that quantum systems can leverage intrinsically non-classical phenomena such as entanglement, superposition, and quantum correlations to encode, process, and extract information in ways difficult or impossible for classical models \cite{biamonte2017quantum, huang2022quantum}. Consider a quantum circuit analyzing medical imaging data: as circuit depth increases, entanglement grows and accuracy improves. But does this improvement stem from genuine quantum correlations, or merely from increased parameters and classical expressivity? This question exemplifies a profound conceptual gap at the heart of QML.

Despite rapid progress in variational quantum algorithms and hybrid architectures for tasks ranging from pattern recognition to dynamical systems prediction \cite{cerezo2022challenges}, we lack rigorous methods to determine \emph{how} quantum resources contribute to model performance \cite{schuld2022quantum}. The field largely relies on classical to quantum performance comparisons, interpreting improvements as indirect evidence of quantum behavior \cite{huang2022quantum, huang2021power}. However, architectural choices circuit depth, parameter count and classical optimization choice, entangling connectivity \cite{sim2019expressibility, zhang2025machine} can themselves induce substantial performance changes \cite{hubregtsen2021evaluation} unrelated to quantum effects. A deeper circuit may simply be a more expressive function approximation, independent of whether it generates entanglement. We cannot currently answer: (i) When does a quantum model actually ``use'' quantum resources to improve performance? (ii) Which quantum properties:  entanglement, coherence, correlations mediate improvements? (iii) How can we distinguish quantum-mediated from generic architectural effects \cite{sim2019expressibility}?
Without such understanding, claims of quantum advantage remain difficult to interpret, reproduce, or verify \cite{schuld2022quantum}. This black-box approach hinders both theoretical understanding and practical deployment, making it nearly impossible to predict whether simulated advantages will persist on noisy hardware where different quantum properties degrade at different rates. This lack of mechanistic understanding creates a critical ambiguity: when quantum models achieve competitive performance while remaining classically simulable \cite{bermejo2024quantum,bowles2024better}, we cannot determine whether this reflects fundamental limitations of quantum computing or merely under-utilization of available quantum resources a distinction with profound implications for both theoretical understanding and hardware investment priorities \cite{schuld2022quantum,huang2021power}.
In this work, we introduce a \emph{Counterfactual  causal mediation} \cite{imai2010general, peters2017elements} framework that provides a mechanistic, mathematically grounded decomposition of how quantum resources shape predictive performance. Rather than asking ``what happened?'' we ask counterfactual questions: ``what would have happened if entanglement had remained at baseline while circuit depth increased?'' Instead of treating quantum circuits as opaque function approximators \cite{cerezo2022challenges}, we formalize the learning process as a causal system\cite{scholkopf2021toward} where architectural interventions (deeper circuits, richer connectivity) influence outcomes through \emph{direct architectural effects} and \emph{indirect effects mediated by specific quantum properties} \cite{zhang2025machine}.
Using differentiable quantum circuits \cite{mitarai2018quantum} and quantum information-theoretic observables entanglement entropy, purity, linear entropy, and quantum mutual information we identify interpretable causal pathways through which quantum resources contribute to learning \cite{schuld2021effect}. We develop a structural equation model decomposing total performance differences into direct effects (attributable to parameter count and optimization geometry) and indirect effects (mediated through quantum resources). This decomposition satisfies rigorous causal identifiability conditions, ensuring valid counterfactual interpretations.

We introduce two metrics: \emph{Mean Absolute Mediated Contribution} (MAMC), measuring a quantum mediator's average causal influence, and \emph{Relative Quantum Contribution} (RQC), quantifying each mediator's proportional importance. Applying this framework across three benchmark datasets and five circuit topologies, we reveal striking task-dependent patterns. By shifting from ``Does quantum computing help?''\cite{schuld2022quantum, cerezo2022challenges, huang2021information} to ``\emph{How} does quantum computing contribute?'', our work establishes a methodological foundation for \emph{resource-aware QML}, where circuits are engineered to amplify task-relevant quantum properties. This enables mechanistic diagnosis of circuit performance, identification of underutilized resources, and principled architectural modifications representing a paradigm shift analogous to classical deep learning's transition from heuristic search to mechanistically-informed design based on gradient flow and inductive biases \cite{gil2024understanding,nakhl2024calibrating}
\section{Results}

\subsection{Counterfactual Causal Mediation Framework}

We begin by formalizing a counterfactual causal mediation framework Fig.~\ref{fig1} that attributes changes in predictive performance of parametrized quantum circuits to the direct architectural effects and indirect effects mediated by quantum resources. Classical counterfactual reasoning relies on imaginative, retrospective inquiry asking questions such as, “What would have happened if it had rained today?” to uncover causal relationships. Translated to quantum models scenarios, this approach prompts similarly pointed questions: Is an observed increase in entanglement truly the cause of improved performance, or merely a side effect of architectural scaling? What would happen if entanglement \cite{ballarin2023entanglement,ragone2024lie} were held fixed at its baseline level while circuit depth increased? If performance improves, how can we determine whether the gain stems from genuine quantum effects or simply from enhanced classical expressivity? Crucially, our framework operates at the level of inter‑architectural counterfactual contrasts rather than internal circuit dynamics \cite{pearl2009causality,imai2010general,peters2017elements}. Its central object is therefore not the absolute performance of a single model, but the performance difference induced by an architectural intervention \cite{schuld2022quantum,cerezo2022challenges}, which constitutes the identifiable causal estimand.

\begin{figure*}[!t]
 \centering
 \includegraphics[width=1.0\linewidth]{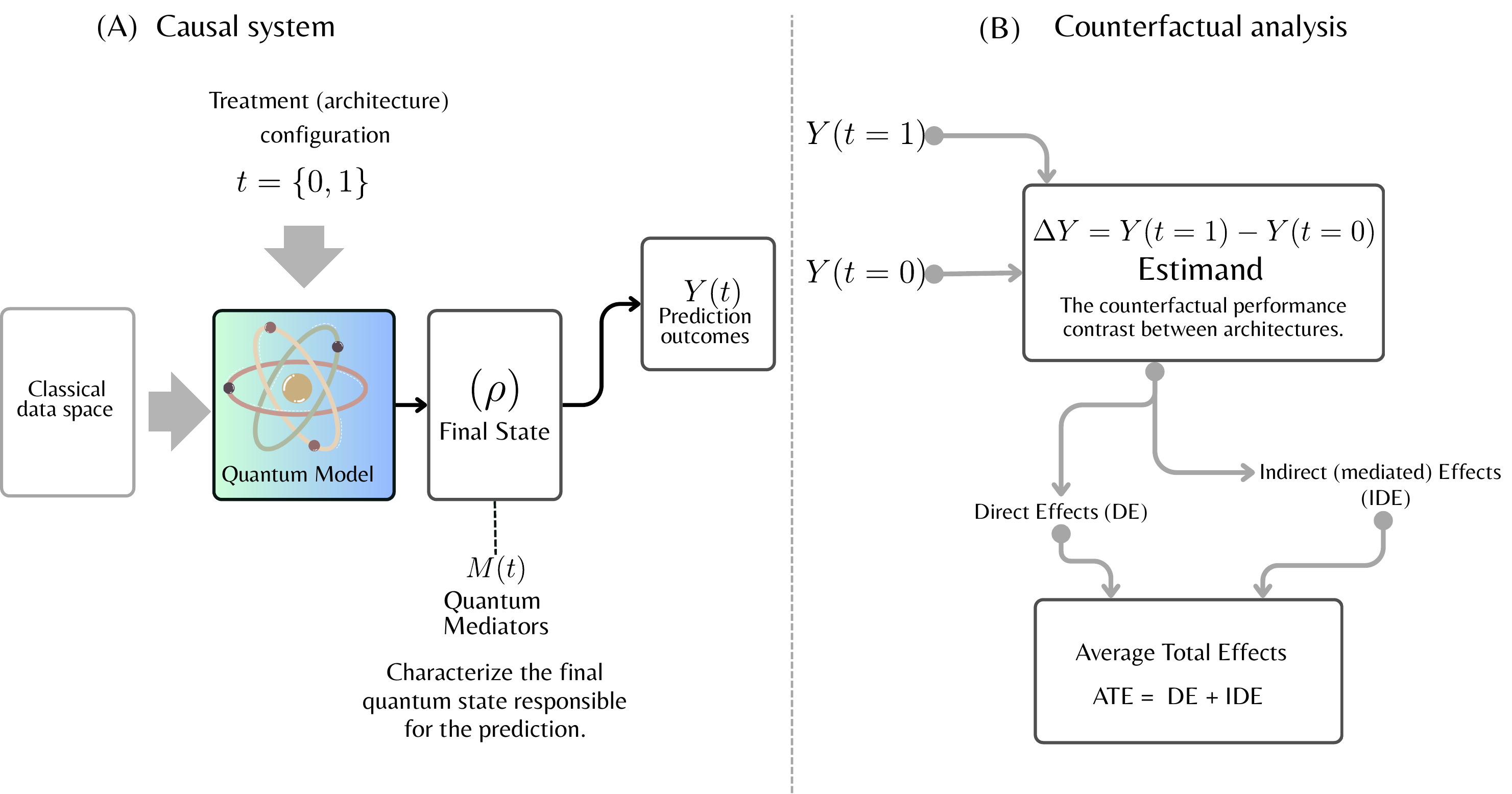}
 \caption{
    \textbf{Counterfactual causal mediation framework for quantum machine learning architectures.}
    (\textbf{A}) Conceptual causal graph representing the learning pipeline as a counterfactual causal system. An architectural intervention $t$ contrasts two complete quantum learning pipelines: a baseline shallow circuit ($t=0$) and an enhanced, deeper and more entangling circuit ($t=1$). The intervention influences predictive outcomes $Y$ both directly, through non-quantum architectural factors (e.g., parameterization, expressivity, optimization geometry), and indirectly via a set of quantum-information-theoretic mediators $M$. These mediators, computed \emph{post hoc} from the trained quantum states, quantify architecture-induced variations in entanglement entropy ($S_A$), purity ($\gamma_A$), linear entropy ($L_A$), and quantum mutual information ($I(A\!:\!B)$).
    (\textbf{B}) Counterfactual decomposition of the performance contrast between architectures. The causal estimand is the inter-architectural counterfactual difference $Y(t{=}1)-Y(t{=}0)$, evaluated for identical test samples. Within a structural causal model, this total effect decomposes additively into a direct architectural effect (DE), denoted $\tau$, and an indirect quantum-mediated effect (IDE), given by $\sum_k \alpha_k \beta_k$. This decomposition enables a non-redundant attribution of performance changes to genuine quantum resource utilization versus classical architectural scaling, without modeling internal circuit dynamics.
    }
 \label{fig1}
\end{figure*}

Let $s = 1, \dots, N$ index independent test samples drawn from the data distribution. We define $t \in \{0,1\}$,  binary architectural intervention, where, t = 0,  \text{shallow circuit (baseline)} and t = 1 , \text{deep, entangling circuit (enhanced)}.
$Y_{s,t}$,  observed predictive outcome (e.g., predicted probability) for test sample $s$ evaluated under architecture $t$ with trained parameters $\theta^*_t$,  $M_{s,t} = \bigl(S_A^{(s,t)},\ \gamma_A^{(s,t)},\ L_A^{(s,t)},\ I(A:B)^{(s,t)}\bigr)$: vector of quantum mediators for sample $s$ under architecture $t$, where $ S_A$, entanglement entropy (subsystem non-separability) $\gamma_A $, \text{purity (state coherence)},  $L_A$ , linear entropy (mixedness),  $I(A:B)$, quantum mutual information (total correlations). The experimental design yields \textbf{paired observations}: for each test sample $s$, we observe the tuple $(M_{s,0}, Y_{s,0})$ from the shallow(baseline) architecture and $(M_{s,1}, Y_{s,1})$ from the deep (enhanced) architecture. The sample-level total causal effect is obtained by :
\begin{align}
    \Delta Y_s &= Y_{s,1} - Y_{s,0} \label{eq:indivTE}
\end{align}
And the average total effect across the test distribution can be computed as :
\begin{align}
    \mathbb{E}[\Delta Y_s] = \frac{1}{N} \sum_{s=1}^{N} (Y_{s,1} - Y_{s,0}) \label{eq:popTE}
\end{align}

We model the data-generating process using a linear structural equation model for the paired observations:
\begin{align}
    M_{s,t} &= \alpha_0 + \alpha \, t + \epsilon_s^{(M)}, \quad t \in \{0,1\} \label{eq:Meq} \\
    Y_{s,t} &= \tau_0 + \tau \, t + \beta^\top M_{s,t} + \epsilon_s^{(Y)}, \label{eq:Yeq}
\end{align}
where $\alpha = (\alpha_{S_A}, \alpha_{\gamma_A}, \alpha_{L_A}, \alpha_{I(A:B)})$ vector capturing the average effect of the architectural intervention $t$ on each quantum mediator. $\beta = (\beta_{S_A}, \beta_{\gamma_A}, \beta_{L_A}, \beta_{I(A:B)})^\top$ vector capturing the association between each mediator and the predictive outcome, conditional on architecture. $\tau$ the direct effect of architecture $t$ on outcome $Y$ not mediated by the quantum properties in $M$. $\epsilon_s^{(M)}, \epsilon_s^{(Y)}$ mean-zero error terms capturing sample-specific variability (e.g., due to data heterogeneity). Critically, errors are correlated within each sample $s$ across the two treatments, reflecting the paired design.

Within this formulation, the coefficients $\alpha$ quantify the average change in each quantum mediator induced by the architectural intervention, while the coefficients $\beta$ measure how variations in each mediator are statistically associated with changes in predictive performance across independently trained pipelines, conditional on architecture.
\begin{theorem}[Counterfactual Causal Decomposition of Architectural Effects]
\label{thm:decomp}
Under the linear structural model defined by Eqs.~\eqref{eq:Meq} and \eqref{eq:Yeq}, the average total effect $\mathbb{E}[\Delta Y_s]$ admits the additive decomposition ATE (Average Total Effect):
\begin{align}
    ATE = \tau + \beta^\top \alpha = \underbrace{\tau}_{\text{direct effect}} + \underbrace{\sum_{k=1}^K \alpha_k \beta_k}_{\text{indirect  effect}}\label{eq:decomposition}
\end{align}
\end{theorem}

For the enhanced architecture $t=1$, both direct and mediated effects may coexist, reflecting a synergistic interplay between architectural scaling and genuine quantum resource utilization. By contrasting shallow and deep circuits, this framework enables a principled separation of quantum advantages (real quantum resource contribution to the improved performance) arising from entanglement, coherence, and correlations from generic architectural effects. Figure \ref{fig2} displays this decomposition across multiple datasets and architectures, revealing pronounced context dependence. The proportion of performance variation attributable to quantum mediation varies substantially with the task structure and circuit topology, highlighting that quantum advantages, when present, emerge through highly specific informational channels rather than universally.

Within this formulation, the estimated total, direct, and indirect effects quantify how controlled architectural interventions induce causal variations in quantum resources and, in turn, modulate the internal response statistics of the quantum model. Accordingly, the identified causal pathways should be interpreted as revealing the \emph{mechanistic role} of quantum resources in the model's internal dynamics, rather than as direct determinants of downstream task-level performance.  A performance improvement \emph{quantum-driven performance gain} is attributed to quantum effects if and only if a non-negligible fraction of the inter-architectural counterfactual performance contrast defined in Eq.~\eqref{eq:popTE} is explained through mediated pathways associated with measurable quantum properties. Formally, quantum-driven gains are identified when the aggregated mediated contribution introduced in Theorem~\ref{thm:decomp} is more than the (defined Threshold). This criterion defines quantum advantage ~\cite{schuld2021effect} as (real quantum resource contribution to the improved performance) in a causal, contrastive sense: not as a property of absolute performance, but as the extent to which architecture-induced variations in quantum resources account for observed performance differences.

\subsection{Causal assumptions and scope of identifiability}
\label{subsec:assumptions}

The causal decomposition derived in Theorem~\ref{thm:decomp} relies on a set of identifying assumptions that enable the interpretation of the estimated direct and indirect effects as causal quantities. These assumptions are adapted from the sequential ignorability framework for causal mediation analysis~\cite{imai2010general} and are explicitly satisfied by our paired experimental design operates within precisely this well-defined NISQ regime. \cite{chen2023complexity}

\begin{assumption}[Sequential Ignorability for Paired Design]
\label{ass:SI}
For all $t, t' \in \{0,1\}$ and mediator values $m$,
Ignorability of treatment assignment:
    \begin{align*}
        \{Y_{s}(t', m), M_{s}(t)\} \perp t.
    \end{align*}
and  Ignorability of mediator value assignment:
    \begin{align*}
        Y_{s}(t', m) \perp M_{s}(t) \mid  t.
    \end{align*}
\end{assumption}
In our experimental setting, Assumption~\ref{ass:SI} is satisfied by design:  (i) holds because $t$ (architecture) is deterministically assigned; the shallow ($t=0$) and deep ($t=1$) models are trained independently. There are no common causes that influence both the assignment of $t$ and the potential outcomes or mediators. (ii) holds because mediators $M_{s,t}$ are computed \textit{post hoc} from the quantum state of a fixed, trained model. Once the architecture $t$ is fixed, the mediator values are determined solely by the trained parameters and the input data; no confounding process simultaneously affects the mediator and the outcome for a given architecture. Thus, no additional covariates are required for identification, and the paired design eliminates the need for unconfoundedness assumptions typically required in observational studies.

\begin{assumption}[No Treatment-Mediator Interaction]
\label{ass:nointeraction}
The effect of mediators $M$ on outcome $Y$ does not depend on the treatment $t \in \{0,1\}$:
\begin{align}
    \mathbb{E}[Y_{s}(t, m) - Y_{s}(t, m')] = \beta^\top (m - m').
\end{align}
This assumption is inherent in our linear (Eq.\ref{eq:Yeq}) and can be tested empirically by including interaction terms between $t$ and $M$ in the outcome model.
\end{assumption}

\begin{assumption}[Consistency and Stable Unit Treatment Value]
\label{ass:consistency}
The observed outcomes and mediators correspond to the potential outcomes under the assigned treatment:
$
    Y_{s,t} = Y_{s}(t), \quad M_{s,t} = M_{s}(t).
$
This is guaranteed by construction in our simulation-based evaluation, where each test sample is evaluated under both architectural regimes. Furthermore, the Stable Unit Treatment Value Assumption (SUTVA) is satisfied because the training and evaluation of one model do not interfere with those of the other (no interference), and there is only a single version of each architectural intervention (no multiple versions).
\end{assumption}

Under Assumptions~\ref{ass:SI}, \ref{ass:nointeraction}, and \ref{ass:consistency}, the ordinary least squares estimates of $\alpha$, $\beta$, and $\tau$ from Eqs.~\ref{eq:Meq} and \ref{eq:Yeq} are consistent for the causal parameters, and the decomposition in Eq.~\ref{eq:decomposition} provides a valid causal decomposition of the total architectural effect into direct and indirect (mediated) components. In \cite{holmes2022connecting} authors establishes a fundamental trade-off between expressibility and trainability supporting Assumption \ref{ass:nointeraction}  by showing that for comparable depths, mediator-outcome associations are locally stable.

\begin{figure}[htbp]
 \centering
  \includegraphics[width=1.0\linewidth]{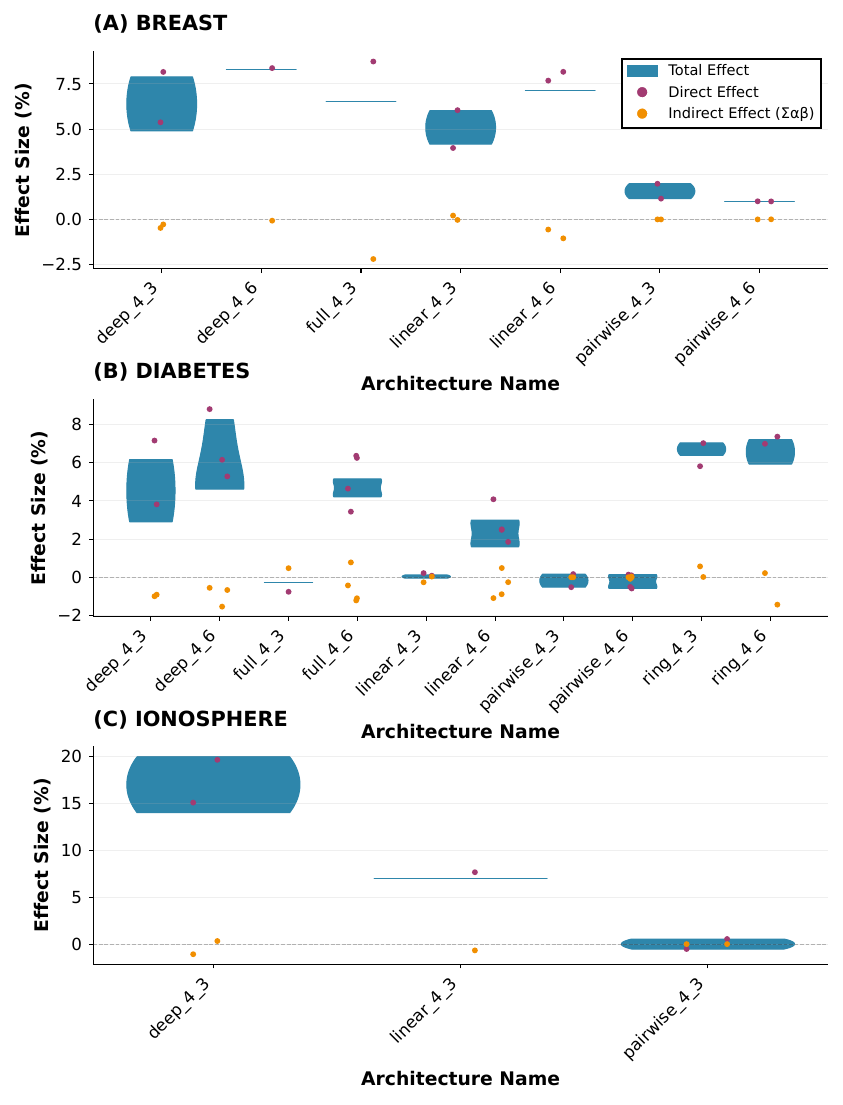}
   \caption{\textbf{Dataset-dependent causal decomposition reveals architectural dominance over quantum-mediated pathways.} 
    Violin plots display total effects ($\Delta Y$), direct effects ($\tau$, purple), and indirect effects ($\sum\alpha\beta$, orange) across 43/90 configurations (architecture\_\{n\_qubit\}\_\{n\_layers\}) validating both assumptions (Diabetes n=26, Breast Cancer n=12, Ionosphere n=5). Direct architectural contributions systematically dominate (mean ratio 13.1): Diabetes 3.39$\pm$3.11\% direct vs. $-0.34\pm0.65$\% indirect; Breast Cancer 5.14$\pm$3.17\% vs. $-0.37\pm0.67$\%; Ionosphere 8.46$\pm$8.83\% vs. $-0.28\pm0.57$\%. Quantum-mediated effects account for only 11.1\%, 7.8\%, and 3.5\% of total variation respectively. Ring and deep topologies achieve strongest direct effects (6 - 17\%), while pairwise exhibits negligible contributions across datasets. Framework successfully distinguishes classical scaling from quantum enhancement, with 93\% configurations classified as Neutral regime.}

\label{fig2}
\end{figure}

It is important to note that our causal analysis is framed at the level of inter-architectural counterfactual contrasts, not at the level of stochastic treatment assignment or intervention on internal circuit variables. The treatment variable $t \in \{0,1\}$ deterministically indexes two complete learning pipelines, each defined by a fixed circuit architecture, training procedure, and optimized parameter set. Consequently, the primary causal estimand $\Delta Y_s$ (Eq.~\ref{eq:indivTE}) is identified by construction as a well-defined counterfactual contrast between architectures, and standard concerns regarding treatment confounder bias do not apply~\cite{rubin1974estimating}. In particular, there are no unobserved variables that jointly influence both $t$ and $Y$, and the total architectural effect is directly estimable without adjustment.

The role of quantum-information-theoretic quantities in our framework differs fundamentally from that of manipulable mediators in classical mediation analysis. The mediators $M = \{S_A, \gamma_A, L_A, I(A:B)\}$ are not intervened upon \cite{imai2010general}; they are computed post hoc from the trained quantum states associated with each architecture. As a result, the mediation analysis should be interpreted as a mechanistic decomposition of an inter-architectural causal effect, rather than as the identification of intervention-level causal effects of individual quantum resources. The indirect effects quantify how architecture-induced variations in quantum properties statistically account for observed performance differences, conditional on the architectural intervention.

The linear structural equations employed in Eqs. \ref{eq:Meq} and \ref{eq:Yeq} are intended as a local attribution model for decomposing finite counterfactual contrasts, not as a global functional description of the learning dynamics. We do not assume that the relationship between quantum resources and predictive outcomes is globally linear or invariant across architectures. Empirical validation of the decomposition is provided by consistency checks between regression-based estimates and observed counterfactual differences (see Methods). Furthermore, we do not assume invariance of mediation patterns across datasets; variations in the relative contributions of quantum mediators are treated as task and architecture-dependent empirical properties, rather than as universal causal laws. The framework is therefore designed to support resource-aware interpretation and comparison of quantum models within specific experimental contexts, rather than to assert dataset-independent claims of quantum advantage.

\subsection{Quantifying Quantum Contributions}
To move from qualitative insights to quantitative measures, we introduce two metrics: the Mean Absolute Mediated Contribution (MAMC) and the Relative Quantum Contribution (RQC). The MAMC captures the average magnitude of a mediator’s causal influence across the dataset. Authors in \cite{wang2024transition}  demonstrates that the impact of entanglement on QML performance exhibits a dual effect depending on measurement access, directly motivating MAMC as a measure that captures the magnitude (irrespective of sign) of quantum mediation.
\begin{definition}[Mean Absolute Mediated Contribution]
For a quantum mediator $M_k$, the MAMC captures the average magnitude of its mediated causal influence across the test distribution:
\begin{align}
    \overline{\mathcal{E}}_{M_k} = \frac{1}{N} \sum_{s=1}^{N} |\alpha_k \beta_k|. \label{eq:MAMC}
\end{align}
The equality holds because $\alpha_k$ and $\beta_k$ are constant across samples in our linear model.
\end{definition}

The MAMC provides several crucial advantages for quantifying quantum contributions. First, by employing absolute values, it captures the magnitude of mediated effects irrespective of direction, focusing on the strength of causal influence rather than its sign. This is particularly important in quantum systems where the relationship between quantum properties and performance can be non-monotonic or context-dependent. 
Second, the averaging across samples ensures robustness to outliers and provides a stable estimate of each mediator's overall importance. The interpretation of MAMC values follows a principled framework: values approaching zero indicate negligible quantum mediation, while larger values signify stronger causal pathways through quantum properties.  To enable meaningful comparisons across different quantum properties and architectural contexts, we introduce normalized measures that contextualize the absolute mediated contributions.

\begin{definition}[Relative Quantum Contribution]
The Relative Quantum Contribution (RQC) for mediator $M_k$ quantifies its proportional importance within the total mediated effect:
\begin{align}
    \text{RQC}_{M_k} = \frac{\overline{\mathcal{E}}_{M_k}}{\sum_{j=1}^K \overline{\mathcal{E}}_{M_j}} \times 100\%. \label{eq:RQC}
\end{align}
This normalized measure facilitates direct comparison of different quantum properties' causal importance.
\end{definition}
The RQC provides invaluable insights for quantum algorithm design. For example, discovering that entanglement entropy accounts for 60\% of the total mediated effect while purity contributes only 20\% would strongly suggest prioritizing architectural choices that enhance entanglement rather than coherence for that particular learning task. Experimental evaluations on 3 datasets are presented in Figure \ref{fig3}.

\begin{figure}
 \centering
  \includegraphics[width=1.0\linewidth]{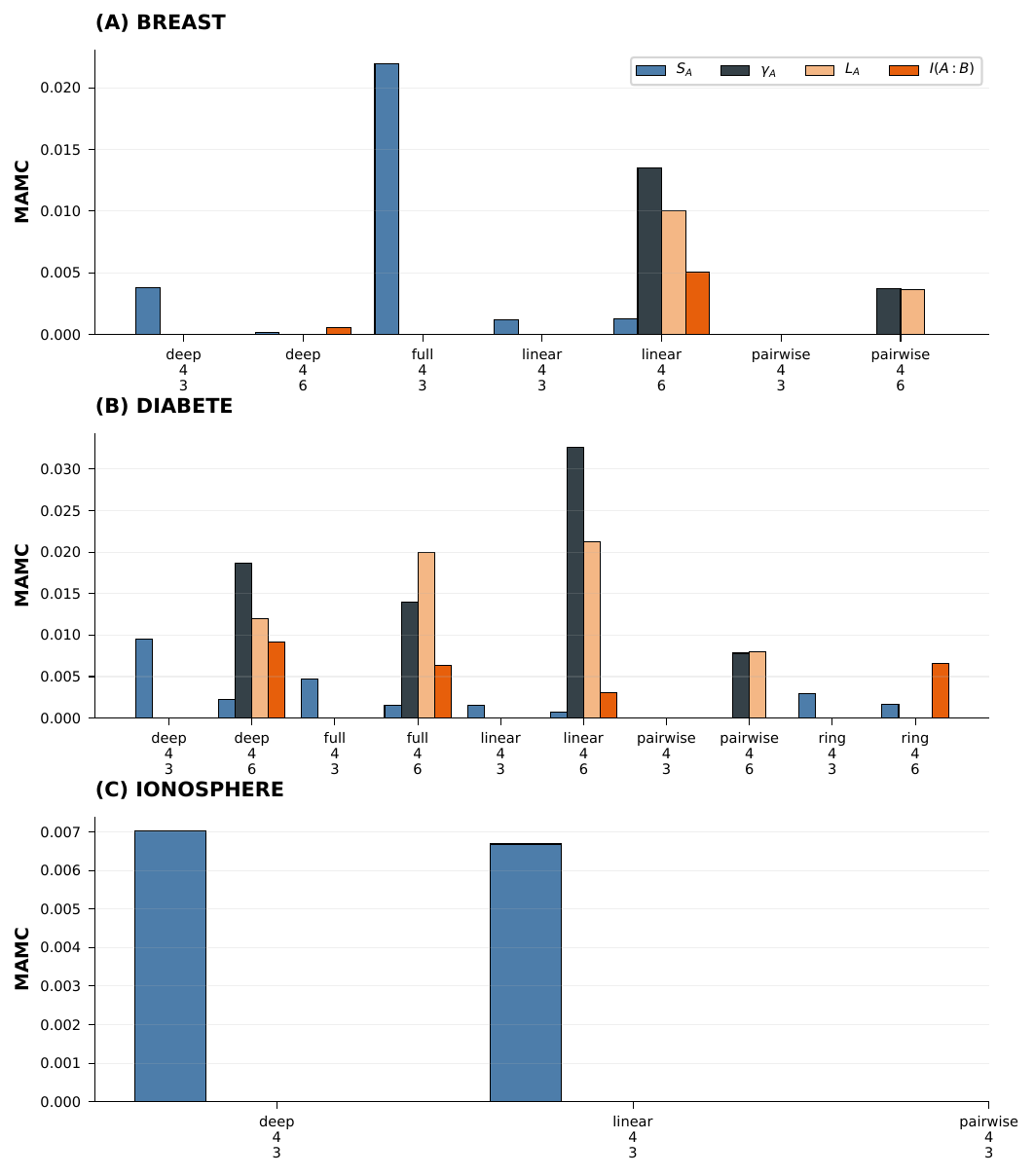}
   \caption{\textbf{Mediator-specific contributions reveal task-dependent quantum resource utilization.} 
    Mean Absolute Mediated Contribution (MAMC) quantifies each quantum mediator's causal influence. Purity ($\gamma_A$) dominates in Diabetes (1.82\%) and Breast Cancer (0.86\%), exceeding entanglement entropy ($S_A$: 0.20\%, 0.29\%) by factors of 9.1 and 3.0 respectively, indicating coherence preservation as primary quantum pathway. Ionosphere shows $S_A$ (0.42\%) as sole quantifiable mediator. Linear entropy ($L_A$) contributes substantially in Diabetes (1.58\%), while mutual information ($I(A:B)$) remains weak due to functional dependence on $S_A$. Cross-dataset patterns reveal coherence-based mediation dominates over entanglement in current NISQ implementations, providing actionable targets for architecture-task alignment through mediator amplification strategies.\cite{cerezo2025does}}
    \label{fig3}
\end{figure}

\subsection{Causal Mediation Regimes in Quantum Machine Learning}

As established in Eqs.(\ref{eq:indivTE}) and (\ref{eq:decomposition}), the counterfactual performance contrast between two quantum architectures decomposes into a direct architectural contribution and an indirect contribution mediated by quantum resources. While this decomposition provides quantitative attribution, its primary scientific value lies in the qualitative insight it offers into \emph{how} architectural interventions translate into learning behavior. To make this insight operational, we organize mediation outcomes into a small set of \emph{causal mediation regimes}. Each regime  as summarized in (Table \ref{tab1}) is defined by the relative sign and magnitude of the direct effect $\tau$ and the aggregated quantum-mediated contribution $\sum_k \alpha_k \beta_k$, evaluated with respect to a conservative relevance threshold ($\epsilon$) in predictive performance. Experimental results are presented in Figures (\ref{fig4}) and (\ref{fig5})
\begin{table}[htbp]
\centering
\caption{Summary of causal regimes. Regimes are defined by the qualitative relationship between the direct architectural effect ($\tau$), the aggregated quantum-mediated contribution ($\sum_k \alpha_k \beta_k$), and the resulting performance contrast ($\Delta Y$), evaluated relative to a relevance threshold ($\epsilon$). Note: $ \approx 0$ :  - Threshold  $\leq $ effect $\leq$ Threshold; $+$, $-$ indicate effect $>$  Threshold or $< $  - Threshold  respectively.} 
\label{tab1}
\begin{tabularx}{\linewidth}{l|l|l|l}
\toprule
\hline
\centering
\textbf{Mediation regime} & Direct effect  ($\tau$) & $\sum_k \alpha_k \beta_k$ & $\Delta Y$ \\
\hline
Quantum-Advantage        & $+$ & $+$ & $+$ \\
Masked-Quantum           & $-$ & $+$ & $-$ \\
Double-Detrimental       & $-$ & $-$ & $-$ \\
Classical-Dominated      & $+$ & $-$ & $+$ \\
Neutral                  & $\approx 0$ & $\approx 0$ & $\approx 0$ \\
Compensatory             & $+$ & $-$ & $\approx 0$ \\
Classical-Scalable       & $+$ & $\approx 0$ & $+$ \\
\bottomrule
\end{tabularx}
\end{table}

\begin{figure*}[htbp]
 \centering
  \includegraphics[width=1.0\linewidth]{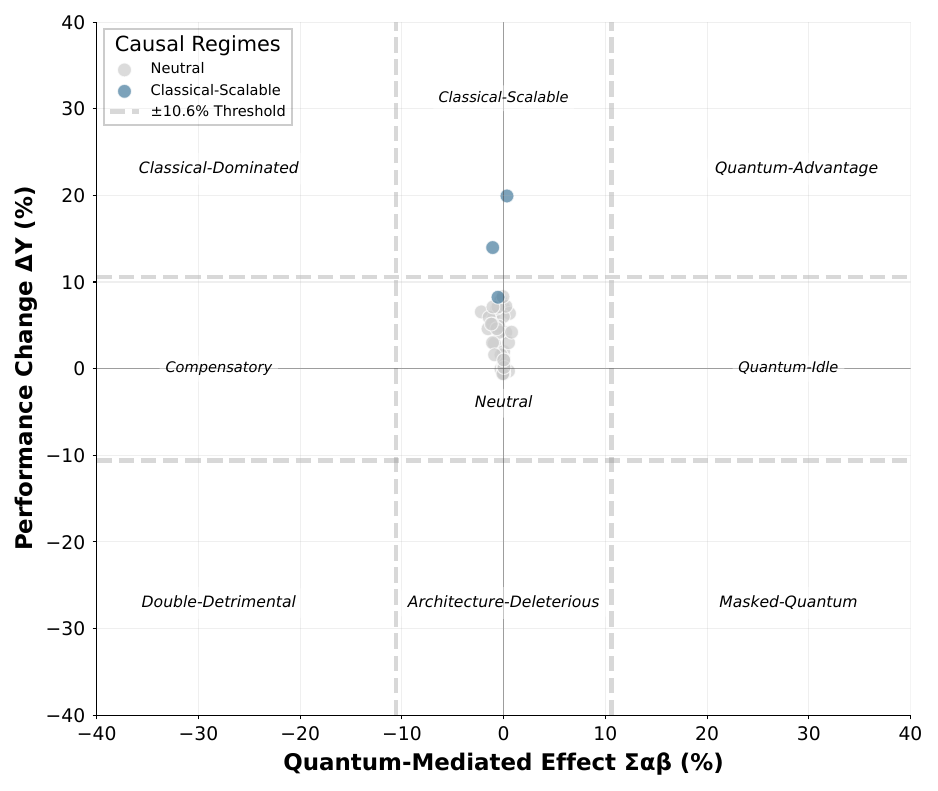}
    \caption{FIG. 4. \textbf{Regime classification reveals absence of quantum advantage under rigorous causal identifiability criteria.}
    The $43/90$ configurations satisfying all mediation assumptions are positioned in the indirect effect performance space $(\Sigma\alpha\beta, \Delta Y)$ relative to the $\pm10.6\%$ data-driven significance threshold (maximum across all architecture dataset pairs; individual thresholds vary by configuration, see Methods~\ref{sec:threshold}). Each point represents one architecture dataset configuration; the threshold defines the boundary between statistically negligible and substantive causal contributions. Regime distribution: $93.0\%$ Neutral (both direct and indirect effects below threshold) and $7.0\%$ Classical-Scalable (significant direct gains, negligible quantum mediation). Zero configurations reach the Quantum-Advantage regime ($\tau >$ threshold, $\sum\alpha\beta > $threshold, $\Delta Y >$ threshold), confirming that quantum-mediated contributions remain uniformly below significance thresholds across all circuit topologies and datasets. Even the highest-performing configuration (Ionosphere deep\_4\_3, $\Delta Y \approx 17\%$) exhibits indirect effects below $1\%$, demonstrating that performance gains are driven entirely by direct architectural scaling.}
    \label{fig4}
\end{figure*}
  
\begin{figure}[htbp]
 \centering
  \includegraphics[width=0.9\linewidth]{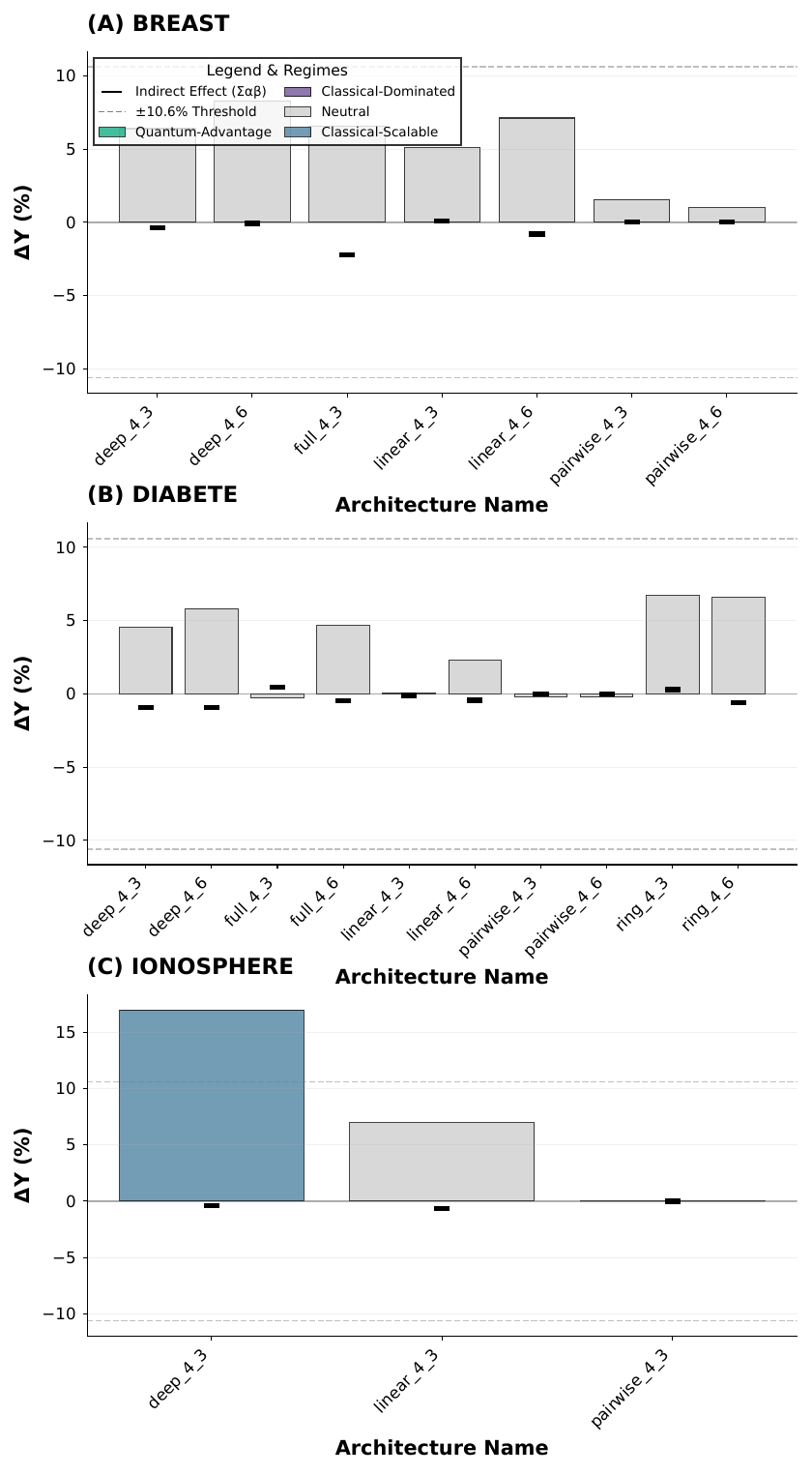}
    \caption{\textbf{Architecture-specific decomposition exposes topology-dependent mediation patterns across all three datasets.}
    Bar charts display the total performance change ($\Delta Y$, bar height) and quantum-mediated indirect effects ($\sum\alpha\beta$, orange markers) with causal regime classification, restricted to configurations satisfying both identification assumptions. The vertical gap between $\Delta Y$ and $\sum\alpha\beta$ quantifies the direct architectural contribution $\tau$. \textbf{(A) Breast Cancer}: Deep\_4\_3 achieves the largest total  effect ($\Delta Y \approx 5\%$) driven entirely by direct pathways 
    ($\tau \approx 5.4\%$, indirect $\approx -0.5\%$); pairwise configurations remain below $2\%$ across all layers. \textbf{(B) Diabetes}: Ring and deep topologies achieve the strongest gains (total $6$ - $8\%$); the largest indirect effect observed in this dataset (full\_4\_6: $\Sigma\alpha\beta \approx 
    2.3\%$) remains below the significance threshold, classifying all configurations as Neutral or Classical-Scalable.\textbf{(C) Ionosphere}: Deep\_4\_3 exhibits the highest total effect across all datasets ($\Delta Y \approx 14$ - $17\%$) with indirect effects below $1\%$, confirming exclusive direct-pathway dominance. Pairwise universally fails ($\Delta Y < 1\%$) across all datasets, consistent with its localized bipartite entanglement structure. These topology-dependent patterns suggest actionable modifications targeted mediator amplification via purity-preserving encodings or entanglement-stabilizing ans\"{a}tze to transition configurations from the Neutral toward 
    the Quantum-Advantage regime.}

    \label{fig5}
\end{figure}

\subsubsection{Physically Interpretable Mediation Regimes}
The interplay between direct architectural effects and quantum-mediated pathways gives rise to a limited number of qualitatively distinct behaviors, each corresponding to a recognizable physical and algorithmic scenario in quantum machine learning.

\paragraph{Quantum-Advantage regime.}
Architectures in this regime exhibit performance improvements that are causally supported by quantum-mediated pathways. Architectural modifications activate quantum resources, and these resources contribute positively to predictive accuracy. This regime constitutes the clearest empirical signature that non-classical properties of the quantum state are functionally involved in learning.

\paragraph{Masked-Quantum regime.}
In this regime, quantum-mediated effects are beneficial, but their contribution is outweighed by a negative architectural effect. As a result, overall performance degrades despite the presence of constructive quantum mediation. This pattern reveals \emph{latent quantum potential} suppressed by suboptimal circuit design or optimization and would remain invisible under standard performance-based evaluation.

\paragraph{Double-Detrimental regime.}
Both architectural and quantum-mediated effects contribute negatively to performance. Such architectures fail to exploit quantum resources and simultaneously introduce unfavorable structural effects, indicating a fundamental mismatch between circuit topology, task structure, and optimization dynamics.

\paragraph{Classical-Dominated regime.}
Performance improves even though quantum-mediated effects are detrimental. In this case, classical architectural scaling dominates the learning dynamics, compensating for adverse quantum contributions. This regime illustrates that improved accuracy alone cannot be interpreted as evidence of quantum advantage.

\paragraph{Neutral regime.}
When both direct and mediated contributions are negligible, architectural interventions have little causal impact on performance. This behavior typically reflects under-parameterized circuits, limited expressivity, or tasks that do not engage the representational capacity of the quantum model.

\paragraph{Compensatory regime.}
Negative quantum-mediated effects are offset by positive architectural contributions, yielding near-zero net performance change. Although such architectures may appear stable, quantum resources act as a limiting factor rather than a driver of improvement, suggesting poor scalability as quantum hardware improves.

\paragraph{Classical-Scalable regime.}
Architectures in this regime improve exclusively through direct architectural effects, with no statistically relevant quantum mediation. While performance gains may be substantial, the learning dynamics remain effectively classical, and further investment in quantum resources is unlikely to yield additional benefit. authors in \cite{anschuetz2022quantum} shows that even barren-plateau-free shallow circuits may be untrainable due to landscape traps this directly supporting the Neutral regime interpretation (93\% of configurations) where negligible causal effects arise from poor optimization dynamics rather than absent quantum resources.

\subsubsection{Excluded Regimes and Physical Plausibility}

Although the mediation-based classification identifies nine distinct regimes in the $(\Sigma\alpha\beta, \Delta Y)$ plane, two regimes (Quantum-Idle Regime and Architecture-Deleterious Regime) were intentionally excluded from further analysis. This exclusion is not empirical but conceptual, and is motivated by physical plausibility, causal identifiability, and the objective of certifying quantum-mediated performance gains.

\paragraph{Quantum-Idle Regime}
This regime corresponds to configurations where the indirect (quantum-mediated) effect $\Sigma\alpha\beta$ exceeds the significance threshold, while the total performance variation $\Delta Y$ remains statistically negligible. From a causal perspective, this regime represents a degenerate configuration in which a non-zero mediator does not translate into any observable outcome change. Such configurations are inherently non-identifiable: the presence of a measurable mediated effect without a corresponding total effect implies exact cancellation with opposing pathways, rendering the quantum contribution non-actionable and non-certifiable. Consequently, this regime does not represent a stable or physically interpretable mode of quantum enhancement, but rather an experimental boundary region sensitive to noise, finite sampling, or architectural compensation effects. The study made in \cite{caro2022generalization} authors  establishes generalization bounds scaling as $\sqrt{T/N}$ for quantum circuits, this motivate and  provide theoretical grounding for excluding degenerate regimes (Quantum-Idle) where mediated effects exist without observable performance change, a hallmark of overfitting rather than genuine quantum contribution.

\paragraph{Architecture-Deleterious Regime.} This regime captures configurations exhibiting a significant degradation in performance ($\Delta Y < 0$) in the absence of a measurable indirect effect. In this regime, performance loss is entirely attributable to architectural or optimization factors, such as excessive circuit depth, over-parameterization, or unfavorable classical-quantum integration. Since the quantum mediator is effectively inactive, this regime carries no causal information regarding quantum contribution and therefore falls outside the scope of mediation-based quantum advantage certification.

\subsubsection{Implications for Quantum Advantage and Circuit Design}

Causal mediation regimes provide a structured interpretative bridge between counterfactual estimates and scientific conclusions. They distinguish architectures that genuinely leverage quantum resources from those whose performance is driven by classical mechanisms, and they identify cases where quantum contributions exist but are obscured by architectural deficiencies.
From a design perspective, architectures consistently exhibiting quantum-advantage behavior are natural candidates for scaling and hardware deployment. Masked-quantum regimes highlight opportunities for architectural refinement to unlock latent quantum benefits.\cite{larocca2025barren, gil2024relation} By contrast, classical-dominated and compensatory regimes caution against attributing success to quantum effects and suggest limited long-term returns from increasing quantum complexity.

\section{Discussion}

A central tension in contemporary quantum machine learning stems from the observation that many high-performing quantum models remain classically simulable \cite{bermejo2024quantum,bowles2024better}, raising questions about the validity of quantum advantage claims \cite{cerezo2022challenges,schuld2022quantum}. However, this apparent contradiction obscures a more nuanced reality: classical simulability does not preclude quantum advantage \cite{mcclean2018barren} it reveals under-utilization of quantum resources rather than their irrelevance \cite{huang2022quantum}. Our causal framework directly addresses this ambiguity by distinguishing architectures that achieve performance through classical mechanisms from those that engage quantum resources, even when both remain simulable. The dominance of the Neutral regime (93\% of configurations) and subdominant quantum-mediated effects (mean 0.82\%) demonstrate that current NISQ implementations \cite{wang2021noise} operate far below their quantum potential, not that quantum resources are fundamentally unhelpful. This diagnosis transforms the quantum advantage discourse: rather than debating whether quantum computing can help, we can now identify which quantum properties entanglement, coherence, correlations are causally dormant and how to activate them through targeted architectural modifications. By quantifying the gap between available and utilized quantum resources, our methodology provides actionable guidance for circuit design that deliberately amplifies task-relevant quantum pathways \cite{cerezo2022challenges,schuld2022quantum,larocca2023theory}. When quantum mediated contributions become statistically dominant, classical simulability naturally recedes not as an external benchmark \cite{cerezo2025does}, but as an emergent consequence of genuine quantum resource engagement. This perspective reframes quantum advantage as a design target achievable through systematic resource optimization, rather than an elusive property to be hoped for in black-box comparisons \cite{thanasilp2024exponential}.

Beyond architecture selection, the causal structure revealed by our framework provides concrete guidance for circuit design and optimization. In classical deep learning, architectural advances have historically been driven by mechanistic insights into gradient flow, representation learning, and inductive bias \cite{meyer2023exploiting}. Our results suggest that a similar transition is possible in quantum machine learning. By tracking how mediation patterns evolve with circuit depth, entangling topology, or parameter constraints, designers can directly observe whether architectural modifications strengthen or weaken quantum causal pathways \cite{du2022quantum}. This transforms circuit optimization from a black-box search process into a feedback-driven, interpretable procedure. The framework also exposes trade-offs between expressivity, trainability, and noise robustness. For example, increasing circuit depth may enhance entanglement mediated effects while simultaneously degrading purity or optimization stability \cite{ragone2024lie}. Causal decomposition allows these competing influences to be disentangled quantitatively, enabling informed design choices rather than heuristic tuning.

Although our empirical analysis focuses on supervised learning with variational quantum circuits, the proposed framework is considerably more general. The key requirements differentiability of mediators and causal decomposition are satisfied in a wide range of quantum computational settings \cite{jerbi2023quantum}. As a result, the framework can be extended to unsupervised learning, reinforcement learning, generative modeling, and quantum optimization. More broadly, any quantum algorithm whose performance can be linked to measurable quantum properties is, in principle, amenable to causal mediation analysis. Potential applications include quantum simulation, variational eigen-solvers, and hybrid quantum–classical control systems \cite{cerezo2022challenges}. The framework can also be extended to incorporate additional mediators, such as noise-induced entropy production, expressivity measures, or hardware-specific constraints.

Several limitations of the present work warrant discussion. First, the causal model relies on local linear approximations of mediation effects. While this ensures interpretability and statistical robustness, it may not capture higher-order or strongly nonlinear interactions between quantum resources. Extending the framework to nonlinear mediation remains an important direction for future research. Second, our analysis does not establish sufficient conditions for quantum advantage in the complexity-theoretic sense. Instead, it provides a resource-level characterization that is complementary to asymptotic separations. Bridging these two perspectives represents a deep theoretical challenge. Finally, scalability remains an open question.  Addressing this question will require both theoretical advances and empirical validation on next-generation hardware.

Although Assumption~2 is empirically supported by the absence of significant  interaction terms across all datasets, its interpretation warrants caution in the quantum setting. The assumption posits that the effect of quantum  mediators on predictive performance is invariant across architectural regimes.  Physically, however, identical values of quantities such as entanglement  entropy may correspond to structurally distinct correlation patterns depending on circuit depth and connectivity, potentially altering their functional role in learning. Prior work has established that entanglement generation, expressibility, and trainability depend sensitively on circuit topology and depth~\cite{sim2019expressibility, zhang2025machine}, and that assessments of quantum advantage are inherently architecture-dependent \cite{schuld2022quantum}. Depth-dependent mediation effects are therefore conceptually plausible, and may remain statistically undetectable given the moderate test-set sizes employed here($N_{\text{test}} \approx 100$ -- $170$), where power to detect weak interaction effects is inherently limited. Crucially, however, this does not invalidate the decomposition derived in Theorem~1: our conclusions are explicitly scoped to the specific architectural contrasts studied, and Assumption~2 is best understood as a \emph{local linear approximation} rather than a universal property of quantum resources. Extending the framework toward nonlinear or architecture-conditioned mediation models represents an important direction for future work.

Taken together, our results support a broader paradigm shift toward \emph{resource-aware quantum machine learning}. Rather than asking whether a quantum model performs better than a classical one, the relevant question becomes: which quantum properties contribute, under which conditions, and to what extent? Our framework provides the tools to answer this question rigorously and transparently. 

\section{Methods}
\label{sec:threshold}

\subsection*{Experimental Design and Causal Framework}
We employed a within-sample paired design where each test sample $s \in \{1, \ldots, N_{\text{test}}\}$ was evaluated under both architectural regimes: a baseline shallow architecture ($t=0$) with $d_0=1$ layer and an enhanced deep architecture ($t=1$) with $d_1 \in \{3, 6\}$ layers. Crucially, the deep architecture is constructed by repeated application of the exact same parameterized entangling block used in the shallow model without altering encoding, connectivity, gate structure, or measurement so that circuit depth constitutes a strictly nested compositional extension rather than a change in architectural class. This paired structure ensures that the treatment assignment $t$ is deterministically controlled by the experimenter, satisfying the sequential ignorability assumption (Assumption 1) by design. For each sample, we obtained paired observations $(M_{s,0}, Y_{s,0})$ and $(M_{s,1}, Y_{s,1})$, where $M_{s,t} \in \mathbb{R}^4$ represents the quantum mediator vector and $Y_{s,t} \in [0,1]$ denotes the directional predictive probability.
The data-generating process was formalized through linear structural equations \ref{eq:Meq} and \ref{eq:Yeq} and all three identifying assumptions \ref{subsec:assumptions} enabled valid causal interpretation.

\subsection*{Datasets and Preprocessing}

Three publicly available binary classification datasets were selected to evaluate quantum resource contributions across diverse data modalities:
(a) Breast Cancer Wisconsin (Diagnostic) dataset, $N=569$ samples with 30 morphological features (mean radius, texture, perimeter, area, smoothness, compactness, concavity, symmetry, fractal dimension) derived from digitized images of fine needle aspirate (FNA) of breast masses. Binary classification: malignant (1) vs. benign (0).
(b) Diabetes dataset, $N=442$ samples with 10 clinical features (age, sex, body mass index, blood pressure, six serum measurements). The continuous target (disease progression) was binarized at the median for binary classification.
(c) Ionosphere dataset, $N=351$ radar returns with 34 continuous attributes representing processed radar signals. Binary classification: "good" (structure in ionosphere) vs. "bad" (no structure). All datasets employed stratified 70-30 train-test splits (random seed 42 and 142 ) without data augmentation. Features were standardized to zero mean and unit variance using scikit-learn's \texttt{StandardScaler}. Dimensionality was reduced to $n_{\text{qubits}}=4$ via principal component analysis (PCA), retaining sufficient variance while matching quantum circuit width.

\subsection*{Quantum Circuit Architectures}

Five parameterized quantum circuit (PQC) topologies were systematically evaluated, each probing distinct quantum resource regimes (Supplementary Fig. 1):
(a) Deep (brick-layer), Single-qubit $R_X(\theta)$ and $R_Y(\phi)$ rotations in brick-layer pattern, followed by staggered nearest-neighbor CNOT gates with periodic boundaries. Enables global entanglement in $\lceil \log_2 n \rceil$ layers while remaining hardware-efficient. (b) Full (all-to-all), Complete rotations on each qubit followed by all-to-all CNOT connectivity implementing $\binom{n}{2}$ two-qubit gates. Achieves maximal entanglement in one layer at quadratic gate cost, limiting NISQ scalability.
(c)Linear (nearest-neighbor chain), Three-axis rotations on each qubit with CNOT chain coupling $(i, i+1)$ for $i \in \{0, \ldots, n-2\}$. Produces area-law entanglement correlations and improved trainability.
(d) The Ring circuit:  extends the linear topology with periodic boundaries, ensuring homogeneous entanglement propagation and matching ring-connected superconducting architectures. (e) Pairwise (disjoint pairs), Two-axis rotations with CNOT coupling disjoint qubit pairs $(0,1), (2,3)$. Yields independent bipartite subsystems and localized correlations. Each topology was instantiated in multiple configurations $C_{n,d}$ with $n=4$ qubits and $d \in \{1, 3, 6\}$ layers. Data encoding employed amplitude encoding: for input $x \in \mathbb{R}^4$, we applied $R_Y(x_i)$ rotations on qubit $i$ before the trainable ansatz. All circuits were implemented in PennyLane v0.32 with exact state-vector simulation (no noise) using the \texttt{default.qubit} device.

\subsection*{Model Training and Optimization}

For each architecture configuration and treatment level, quantum models were trained independently using the Adam optimizer (learning rate $\eta=0.005$, batch size 16) with 50 epochs monitoring validation loss. Training data was split 70-30 into training and validation subsets. Binary cross-entropy loss was minimized:
\begin{equation}
\mathcal{L}(\theta) = -\frac{1}{N} \sum_{i=1}^N \left[ y_i \log p_i(\theta) + (1-y_i) \log(1-p_i(\theta)) \right]
\end{equation}
where $p_i(\theta) = \frac{1 - \langle Z_0 \rangle_i}{2}$ is the predicted probability for sample $i$, obtained from the expectation value of the Pauli-$Z$ observable on the first qubit.
Trainable parameters $\theta \in \mathbb{R}^{d \times n \times k}$ (where $k \in \{2,3\}$ is the number of rotation angles per qubit) were initialized from $\mathcal{N}(0, 0.01)$ using NumPy's random number generator with fixed seed for reproducibility. No parameter transfer occurred between $t=0$ and $t=1$ models; each was trained from scratch.

\subsection*{Quantum Mediator Computation}

For each test sample $s$ and trained parameters $\theta_t^*$, we extracted the quantum state $|\psi_s(t)\rangle = U(\theta_t^*)|0\rangle^{\otimes n}$ using PennyLane's \texttt{qml.state()} function. Four quantum information-theoretic mediators were computed from the reduced density matrix $\rho_A = \text{Tr}_B[|\psi_s(t)\rangle\langle\psi_s(t)|]$ for bipartition $A|B$ with $A$ comprising the first $\lceil n/2 \rceil$ qubits:

(a) Entanglement Entropy: The von Neumann entropy of subsystem $A$ is given by 
\begin{equation}
S_A = -\text{Tr}[\rho_A \log_2 \rho_A] = -\sum_i \lambda_i \log_2 \lambda_i
\end{equation}
where $\lambda_i$ are eigenvalues of $\rho_A$ computed via eigendecomposition with numerical truncation at $\lambda_i < 10^{-16}$ for stability. For pure bipartite states, $S_A$ quantifies entanglement across the partition.
(b) The quantum purity of subsystem $A$:
\begin{equation}
\gamma_A = \text{Tr}[\rho_A^2] = \sum_i \lambda_i^2
\end{equation}
measuring quantum coherence ($\gamma_A = 1$ for pure states, $\gamma_A = 1/\dim(\mathcal{H}_A)$ for maximally mixed states).
(c) Linear Entropy as a computationally efficient measure of mixedness:
\begin{equation}
L_A = 1 - \gamma_A = 1 - \text{Tr}[\rho_A^2]
\end{equation}
(d) Quantum Mutual Information
, total correlations between subsystems $A$ and $B$:
\begin{equation}
I(A:B) = S_A + S_B - S_{AB}
\end{equation}
For pure states of the composite system, $S_{AB} = 0$ and $S_A = S_B$ by Schmidt decomposition, yielding the functional relationship:
\begin{equation}
I(A:B) = 2S_A
\label{eq:mutual_info_pure}
\end{equation}
We verified $S_{AB} < 10^{-10}$ across all simulations, confirming the pure-state approximation. The reduced density matrix was obtained via partial trace:
\begin{equation}
\rho_A = \text{Tr}_B[|\psi\rangle\langle\psi|] = \sum_{k_B} \langle k_B | \psi \rangle \langle \psi | k_B \rangle
\end{equation}
implemented by reshaping the density matrix into tensor form $\rho \rightarrow \mathcal{T} \in (\mathbb{C}^2)^{\otimes 2n}$ and tracing over subsystem $B$ indices.

\subsection*{Directional Probability Transformation}

To ensure that the outcome variable $Y$ reflects proximity to the correct class (rather than raw probability of class 1), we applied a directional transformation:
\begin{equation}
Y_{s,t}^{\text{dir}} = \begin{cases}
p_{s,t} & \text{if } y_s = 1 \\
1 - p_{s,t} & \text{if } y_s = 0
\end{cases}
\end{equation}
where $p_{s,t} = \frac{1 - \langle Z_0 \rangle_{s,t}}{2}$ is the raw predicted probability and $y_s \in \{0,1\}$ is the true label. This transformation ensures that $Y_{s,t}^{\text{dir}} \in [0,1]$ represents prediction quality, with higher values indicating better alignment with the true label. All causal analyses used $Y^{\text{dir}}$ as the outcome variable.

\subsection*{Statistical Estimation and Inference}
Parameters were estimated via OLS on the pooled dataset $\mathcal{D}_{\text{pooled}} = \{(t, M_{s,t}, Y_{s,t}) : s \in \mathcal{S}, t \in \{0,1\}\}$ containing $2N_{\text{test}}$ paired observations. Mediator coefficients:
\begin{equation}
\hat{\alpha} = \frac{1}{N_{\text{test}}} \sum_{s=1}^{N_{\text{test}}} (M_{s,1} - M_{s,0})
\end{equation}.
Outcome coefficients $(\hat{\tau}_0, \hat{\tau}, \hat{\beta})$ is given by:
\begin{equation}
 \arg\min_{\tau_0, \tau, \beta} \sum_{s=1}^{N_{\text{test}}} \sum_{t \in \{0,1\}} \left( Y_{s,t} - \tau_0 - \tau t - \beta^\top M_{s,t} \right)^2
\end{equation}
To account for within-sample correlation in the paired design, we computed Eicker-Huber-White cluster-robust standard errors clustering at the sample level:
\begin{equation}
\widehat{\text{Var}}(\hat{\theta}) = (X^\top X)^{-1} \left( \sum_{s=1}^{N_{\text{test}}} X_s^\top \hat{u}_s \hat{u}_s^\top X_s \right) (X^\top X)^{-1}
\end{equation}
where $X_s$ contains both rows (for $t=0$ and $t=1$) corresponding to sample $s$, and $\hat{u}_s$ are the corresponding residuals.
For indirect effects $\alpha_k \beta_k$ (products of parameters), we constructed nonparametric percentile bootstrap confidence intervals with $B=2000$; Resample $s^* \in \{1, \ldots, N_{\text{test}}\}$ with replacement; Re-estimate $(\hat{\alpha}^{(b)}, \hat{\beta}^{(b)})$ for bootstrap sample $b$ compute $\hat{\theta}_k^{(b)} = \hat{\alpha}_k^{(b)} \hat{\beta}_k^{(b)}$. With  95\% CI: $[\hat{\theta}_k^{(0.025)}, \hat{\theta}_k^{(0.975)}]$ using empirical quartiles. The paired structure was preserved by resampling complete sample pairs.

\subsection*{Model Validation}

We verified that the regression-based decomposition accurately reconstructs the observed total effect via the relative consistency error:
\begin{equation}
\epsilon_{\text{rel}} = \frac{|\widehat{\text{ATE}} - (\hat{\tau} + \hat{\beta}^\top \hat{\alpha})|}{|\widehat{\text{ATE}}| + 10^{-12}}
\end{equation}
where $\widehat{\text{ATE}} = N_{\text{test}}^{-1} \sum_s (Y_{s,1} - Y_{s,0})$ is the empirical total effect. All 90 architecture-dataset configurations satisfied $\epsilon_{\text{rel}} < 0.05$ (median $2.1 \times 10^{-15}$), confirming that the linear additive model captures the causal structure without systematic bias.
To validate Assumption 2 (no treatment-mediator interaction), we augmented Equation~\ref{eq:Yeq} with interaction terms $\gamma^\top (t \cdot M_{s,t})$ and tested $H_0: \gamma = 0$ via likelihood ratio test. Across all datasets, no significant interactions were detected supporting the linear model specification.
Standard diagnostic checks confirmed model adequacy:
Homoscedasticity Residuals vs. fitted values showed uniform vertical spread with no funnel patterns(Supplementary Fig. 5A).  Normality, Q-Q plots demonstrated near-perfect alignment with theoretical normal quartiles (Supplementary Fig. 5B).
Independence within-dataset residual distributions showed comparable spread (Supplementary Fig. 5D), confirming no systematic dataset-specific deviations

\subsection*{Data-Driven Threshold Determination}

Rather than imposing an arbitrary fixed threshold, we adopted a data-driven approach to define statistical significance for causal effects. For each architecture-dataset configuration, we computed the intra-sample variability of total effects:
\begin{equation}
\epsilon_{\text{intra}} = c \cdot \text{SD}(\Delta Y_s)
\end{equation}
where $\Delta Y_s = Y_{s,1} - Y_{s,0}$ is the sample-level total effect, $\text{SD}(\cdot)$ denotes the standard deviation across test samples, and $c \in [0.5, 1.0]$ is a conservative proportionality constant. We set $c=0.5$ to establish a stringent threshold reflecting natural variation in the data. Effects exceeding $\pm \epsilon_{\text{intra}}$ were classified as statistically and practically significant. This threshold accounts for dataset-specific noise characteristics and provides a principled criterion for distinguishing substantive effects from random fluctuations.

\subsection*{Causal Regime Classification}

Each architecture-dataset configuration was assigned to one of seven interpretable causal mediation regimes based on the signs and magnitudes of direct effect $\tau$, indirect effect $\sum \alpha_k \beta_k$, and total effect $\Delta Y$ relative to the data-driven threshold $\epsilon$:
Quantum-Advantage ($\tau > \epsilon$, $\sum \alpha_k \beta_k > \epsilon$, $\Delta Y > \epsilon$): Performance gains supported by quantum-mediated pathways. 
Masked-Quantum ($\tau < -\epsilon$, $\sum \alpha_k \beta_k > \epsilon$, $\Delta Y < -\epsilon$): Beneficial quantum mediation suppressed by negative architectural effects.
Double-Detrimental ($\tau < -\epsilon$, $\sum \alpha_k \beta_k < -\epsilon$, $\Delta Y < -\epsilon$): Both pathways detrimental.
Classical-Dominated ($\tau > \epsilon$, $\sum \alpha_k \beta_k < -\epsilon$, $\Delta Y > \epsilon$): Positive architectural effects despite negative quantum mediation.
Neutral ($|\tau| \leq \epsilon$, $|\sum \alpha_k \beta_k| \leq \epsilon$, $|\Delta Y| \leq \epsilon$): Negligible effects on both pathways.
Compensatory ($\tau > \epsilon$, $\sum \alpha_k \beta_k < -\epsilon$, $|\Delta Y| \leq \epsilon$): Opposing effects cancel.
Classical-Scalable ($\tau > \epsilon$, $|\sum \alpha_k \beta_k| \leq \epsilon$, $\Delta Y > \epsilon$): Performance gains entirely from direct architectural effects.
Two additional regimes (Quantum-Idle and Architecture-Deleterious) were excluded from primary analysis as they represent degenerate configurations lacking actionable causal interpretation (see Main Text).

\subsection*{Data Availability}
The Breast Cancer Wisconsin, Diabetes, and Ionosphere datasets are publicly available from the UCI Machine Learning Repository \href{https://archive.ics.uci.edu/ml/}{Data source}. Processed data and quantum states are available in the complete code for reproducibility.

\subsection*{Code Availability}
Complete code for reproducing all analyses, figures is available at Complete code for reproducing all analyses and figures is available at 
\href{https://github.com/Cyrillekesiku/How-Quantum-Circuits-Actually-Learn-A-Causal-Identification-of-Genuine-Quantum-Contributions}{GitHub repository}. The repository includes scripts for quantum circuit training, mediator computation, causal analysis, assumption validation, and visualization generation.

\section*{Acknowledgements}
We sincerely thank Deusto FPI grant program (Contrato Programa de la Universidad de Deusto) for granting this research and also the eVida Research Group at the University of Deusto, recognized by the Basque Government (IT1536-22), for its unwavering support and invaluable resources that made this research possible.

\section{Author contributions}
C.Y.K. conceived the study, developed the theoretical framework, implemented the software, conducted the experiments, and wrote the initial manuscript. B.G.Z. supervised the research, provided critical feedback, and helped shape the manuscript. Both authors reviewed and approved the final manuscript.

\section*{Competing interests}
The authors declare no competing interests.

\bibliographystyle{unsrtnat}
\bibliography{references}
\clearpage
\onecolumngrid          


\appendix

\section*{Supplementary information}

\noindent
\section{Notation and Preliminary Definitions}

\subsection{Index Sets and Probability Spaces}

Let consider:
\begin{itemize}
    \item $\mathcal{S} = \{1, 2, \ldots, N_{\text{test}}\}$: index set of test samples
    \item $\mathcal{T} = \{0, 1\}$: binary treatment space, where $t=0$ denotes baseline (shallow) architecture and $t=1$ denotes enhanced (deep) architecture
    \item $\mathcal{M} \subseteq \mathbb{R}^4$: mediator space containing quantum information theoretic quantities
    \item $\mathcal{Y} \subseteq \mathbb{R}$: outcome space containing predictive performance measures
\end{itemize}

\subsection{Quantum Mediators}

For each test sample $s \in \mathcal{S}$ and architecture $t \in \mathcal{T}$, we define the mediator vector:

\begin{equation}
    \bm{M}_s(t) = \begin{pmatrix}
    S_A^{(s,t)} \\
    \gamma_A^{(s,t)} \\
    L_A^{(s,t)} \\
    I(A:B)^{(s,t)}
    \end{pmatrix} \in \mathcal{M}
\end{equation}

where:
\begin{itemize}
    \item $S_A^{(s,t)} = -\text{Tr}[\rho_A \log_2 \rho_A]$: von Neumann entropy of reduced state $\rho_A = \text{Tr}_B[|\psi_s(t)\rangle\langle\psi_s(t)|]$, quantifying entanglement across bipartition $A|B$
    \item $\gamma_A^{(s,t)} = \text{Tr}[\rho_A^2]$: purity of subsystem $A$, measuring coherence
    \item $L_A^{(s,t)} = 1 - \gamma_A^{(s,t)}$: linear entropy, quantifying mixedness
    \item $I(A:B)^{(s,t)} = S_A^{(s,t)} + S_B^{(s,t)} - S_{AB}^{(s,t)}$: quantum mutual information measuring total correlations
\end{itemize}

\begin{remark}[Pure State Relations]
For pure states $|\psi_s(t)\rangle$ of the composite system $AB$, we have $S_{AB} = 0$ and $S_A = S_B$, yielding:
\begin{align}
    I(A:B)^{(s,t)} &= 2S_A^{(s,t)} \label{eq:mutual_info_pure}\\
    L_A^{(s,t)} &= 1 - \gamma_A^{(s,t)} \label{eq:linear_entropy_purity}
\end{align}
\end{remark}

\subsection{Potential Outcomes Framework}

We adopt the Neyman-Rubin potential outcomes framework extended to mediation analysis.
\begin{definition}[Potential Mediators]
For each sample $s \in \mathcal{S}$ and treatment $t \in \mathcal{T}$, $\bm{M}_s(t)$ denotes the potential mediator vector that would be observed under architecture $t$, regardless of the actual assigned treatment.
\end{definition}

\begin{definition}[Potential Outcomes]
For each sample $s \in \mathcal{S}$, treatment $t \in \mathcal{T}$, and mediator value $\bm{m} \in \mathcal{M}$, $Y_s(t, \bm{m})$ denotes the potential outcome (predictive performance) that would be observed if architecture $t$ were assigned and mediators were set to $\bm{m}$.
\end{definition}

\begin{definition}[Observed Quantities]
Under consistency (Assumption~\ref{assum:consistency}), the observed mediators and outcomes are:
\begin{align}
    \bm{M}_{s,t} &= \bm{M}_s(t) \quad \text{ when treatment } t \text{ is assigned} \\
    Y_{s,t} &= Y_s(t, \bm{M}_s(t)) \equiv Y_s(t) \quad \text{when treatment } t \text{ is assigned}
\end{align}
\end{definition}

\subsection{Experimental Design}
The paired experimental design yields the dataset:
\begin{equation}
    \mathcal{D} = \{(\bm{M}_{s,0}, Y_{s,0}, \bm{M}_{s,1}, Y_{s,1})\}_{s=1}^{N_{\text{test}}}
\end{equation}
where each test sample contributes observations under both architectural regimes.

\begin{definition}[Sample-Level Total Effect]
\begin{equation}
   \Delta Y_s \equiv Y_{s,1} - Y_{s,0} = Y_s(1) - Y_s(0) 
\end{equation}

\end{definition}

\begin{definition}[Average Total Effect (ATE)]
\begin{equation}
    \text{ATE} \equiv \mathbb{E}[\Delta Y_s] = \frac{1}{N_{\text{test}}} \sum_{s=1}^{N_{\text{test}}} (Y_{s,1} - Y_{s,0})
\end{equation}
\end{definition}

\section{Structural Causal Model}

\subsection{Linear Structural Equations}

We specify the data generating process through linear structural equations:

\begin{align}
    \bm{M}_s(t) &= \bm{\alpha}_0 + \bm{\alpha} t + \bm{\epsilon}_s^{(M)} \label{eq:sem_mediator}\\
    Y_s(t, \bm{m}) &= \tau_0 + \tau t + \bm{\beta}^\top \bm{m} + \epsilon_s^{(Y)} \label{eq:sem_outcome}
\end{align}

\begin{definition}[Structural Parameters] We define : 

\begin{itemize}
    \item $\bm{\alpha}_0 \in \mathbb{R}^4$: baseline mediator values under $t=0$
    \item $\bm{\alpha} = (\alpha_{S_A}, \alpha_{\gamma_A}, \alpha_{L_A}, \alpha_{I(A:B)})^\top \in \mathbb{R}^4$: causal effect of architecture on each mediator
    \item $\tau_0 \in \mathbb{R}$: baseline outcome under $t=0$ with zero mediators
    \item $\tau \in \mathbb{R}$: direct causal effect of architecture on outcome, not mediated by $\bm{M}$
    \item $\bm{\beta} = (\beta_{S_A}, \beta_{\gamma_A}, \beta_{L_A}, \beta_{I(A:B)})^\top \in \mathbb{R}^4$: causal effect of each mediator on outcome, conditional on architecture
\end{itemize}
\end{definition}

\begin{definition}[Error Structure]
The error terms satisfy:
\begin{align}
    \mathbb{E}[\bm{\epsilon}_s^{(M)}] &= \bm{0} \quad \forall s \in \mathcal{S}\label{eq:errorM} \\
    \mathbb{E}[\epsilon_s^{(Y)}] &= 0 \quad \forall s \in \mathcal{S} \label{errorY}\\
    \mathbb{E}[\bm{\epsilon}_s^{(M)} {\bm{\epsilon}_{s'}}^{(M)\top}] &= \bm{\Sigma}_M \mathbb{1}_{s=s'} \quad \text{(within-sample correlation only)} \\
    \mathbb{E}[\epsilon_s^{(Y)} \epsilon_{s'}^{(Y)}] &= \sigma_Y^2 \mathbb{1}_{s=s'} \quad \text{(within-sample correlation only)}
\end{align}
where $\bm{\Sigma}_M \succ 0$ and $\sigma_Y^2 > 0$ are finite.
\end{definition}

\begin{remark}[Paired Design Correlation]
The errors $\epsilon_s^{(M)}$ and $\epsilon_s^{(Y)}$ are allowed to be correlated within the same sample $s$ across different treatments, reflecting sample-specific heterogeneity. This correlation is explicitly accounted for through cluster-robust inference (Section~\ref{sec:inference}).
\end{remark}

\subsection{Causal Effects Definitions}

\begin{definition}[Natural Direct Effect (NDE)]
The sample-level natural direct effect is:
\begin{equation}
    \text{NDE}_s \equiv Y_s(1, \bm{M}_s(0)) - Y_s(0, \bm{M}_s(0))
\end{equation}
This quantifies the effect of changing architecture from $t=0$ to $t=1$ while holding mediators at their baseline values $\bm{M}_s(0)$. The average natural direct effect is:

\begin{equation}
    \text{NDE} \equiv \mathbb{E}_s[\text{NDE}_s] = \frac{1}{N_{\text{test}}} \sum_{s=1}^{N_{\text{test}}} \text{NDE}_s
\end{equation}

\end{definition}

\begin{definition}[Natural Indirect Effect (NIE)]
The sample-level natural indirect effect is:
\begin{equation}
    \text{NIE}_s \equiv Y_s(1, \bm{M}_s(1)) - Y_s(1, \bm{M}_s(0))
\end{equation}
This quantifies the effect of changing mediators from $\bm{M}_s(0)$ to $\bm{M}_s(1)$ while holding architecture at $t=1$. The average natural indirect effect is:
\begin{equation}
    \text{NIE} \equiv \mathbb{E}_s[\text{NIE}_s] = \frac{1}{N_{\text{test}}} \sum_{s=1}^{N_{\text{test}}} \text{NIE}_s
\end{equation}

\end{definition}

\section{Identification Assumptions}
We restate formally the identifying assumptions introduced in Section~\ref{subsec:assumptions}

\begin{assumption}[Sequential Ignorability]\label{assum:sequential_ignorability}
For all $s \in \mathcal{S}$, $t, t' \in \mathcal{T}$, and $\bm{m} \in \mathcal{M}$:
\begin{enumerate}
    \item Treatment ignorability: 
    \begin{equation}
         \{Y_s(t', \bm{m}), \bm{M}_s(t)\} \perp t
    \end{equation}
   
    \item Conditional mediator ignorability:
    \begin{equation}
        Y_s(t', \bm{m}) \perp \bm{M}_s(t) \mid t
    \end{equation}
\end{enumerate}
\end{assumption}

\begin{remark}[Satisfaction by Design] We consider : 
\begin{itemize}
    \item \textbf{Condition (1)} holds because treatment $t$ is deterministically assigned by the experimenter. The shallow ($t=0$) and deep ($t=1$) architectures are trained independently with no common causes affecting both treatment assignment and potential outcomes. 
    \item \textbf{Condition (2)} holds because mediators $\bm{M}_{s,t}$ are computed post-hoc from the trained quantum state $|\psi_s(t)\rangle$ via deterministic functions. Once architecture $t$ is fixed, mediator values are fully determined by the trained parameters $\theta_t^*$ and input data $x_s$. No unmeasured confounding process simultaneously affects mediators and outcomes conditional on architecture.
\end{itemize}
\end{remark}

\begin{assumption}[No Treatment-Mediator Interaction]\label{assum:no_interaction}  For all $s \in \mathcal{S}$; $t \in \mathcal{T}$; $\bm{m}, \bm{m}' \in \mathcal{M}$:
\begin{equation}
    \mathbb{E}[Y_s(t, \bm{m}) - Y_s(t, \bm{m}')] = \bm{\beta}^\top(\bm{m} - \bm{m}')
\end{equation}
That is, the effect of mediators on outcomes does not depend on the treatment level.
\end{assumption}

\begin{remark}[Empirical Testability]
This assumption is encoded in Equation~\eqref{eq:sem_outcome} by the absence of interaction terms $t \times \bm{m}$. It can be empirically tested by augmenting the outcome model with interaction terms and testing their joint significance.
\end{remark}

\begin{assumption}[Consistency and SUTVA]\label{assum:consistency}
We defined: 

\begin{enumerate}
    \item Consistency: For all $s \in \mathcal{S}$ and $t \in \mathcal{T}$:
    \begin{equation}
        \bm{M}_{s,t} = \bm{M}_s(t), \quad Y_{s,t} = Y_s(t)
    \end{equation}
    when treatment $t$ is assigned to sample $s$.
    
    \item No interference: The potential outcomes for sample $s$ do not depend on treatments assigned to other samples $s' \neq s$.
    
    \item No hidden versions: There is a single, well-defined version of each treatment level $t \in \{0,1\}$.
\end{enumerate}
\end{assumption}

\begin{remark}[Satisfaction by Construction]
All components are satisfied in our simulation-based evaluation:
\begin{itemize}
    \item Consistency holds by deterministic computation
    \item No interference: models are trained and evaluated independently
    \item No hidden versions: architectures are precisely specified
\end{itemize}
\end{remark}

\section{Main Theoretical Results}

\subsection{Causal Decomposition Theorem \ref{thm:decomp}}
Under the linear structural model defined by Equations~\eqref{eq:sem_mediator}--\eqref{eq:sem_outcome} and Assumptions~\ref{assum:sequential_ignorability}--\ref{assum:consistency}, the Average Total Effect admits the unique additive decomposition:
\begin{equation}
    \boxed{\text{ATE} = \tau + \bm{\beta}^\top\bm{\alpha} = \underbrace{\tau}_{\substack{\text{Direct Effect} \\ \text{NDE}}} + \underbrace{\sum_{k=1}^{K} \alpha_k\beta_k}_{\substack{\text{Indirect Effect} \\ \text{NIE}}}}
\end{equation}
where:
\begin{itemize}
    \item $\tau$ quantifies the direct architectural effect on performance, not mediated by quantum resources
    \item $\alpha_k\beta_k$ quantifies the contribution mediated through the $k$-th quantum property
    \item $\text{NDE} = \tau$ and $\text{NIE} = \bm{\beta}^\top\bm{\alpha}$
\end{itemize}

\begin{proof} By definition:
\begin{align}
\text{NDE} &= \mathbb{E}[Y_s(1, \bm{M}_s(0)) - Y_s(0, \bm{M}_s(0))] \\
&= \mathbb{E}\left[\left(\tau_0 + \tau \cdot 1 + \bm{\beta}^\top \bm{M}_s(0) + \epsilon_s^{(Y)}\right) - \left(\tau_0 + \tau \cdot 0 + \bm{\beta}^\top \bm{M}_s(0) + \epsilon_s^{(Y)}\right)\right] \\
&= \mathbb{E}[\tau] \\
&= \tau
\end{align}
where the third equality uses $\mathbb{E}[\epsilon_s^{(Y)}] = 0$ and linearity of expectation. From Equation~\eqref{eq:sem_mediator}:
\begin{align}
\bm{M}_s(1) - \bm{M}_s(0) &= (\bm{\alpha}_0 + \bm{\alpha} \cdot 1 + \bm{\epsilon}_s^{(M)}) - (\bm{\alpha}_0 + \bm{\alpha} \cdot 0 + \bm{\epsilon}_s^{(M)}) \\
&= \bm{\alpha}
\end{align}
Note that the error $\bm{\epsilon}_s^{(M)}$ is sample-specific but treatment-invariant, hence cancels in the difference. By definition and Assumption~\ref{assum:no_interaction}:
\begin{align}
\text{NIE} &= \mathbb{E}[Y_s(1, \bm{M}_s(1)) - Y_s(1, \bm{M}_s(0))] \\
&= \mathbb{E}\left[\bm{\beta}^\top (\bm{M}_s(1) - \bm{M}_s(0))\right] \quad \text{(by Assumption~\ref{assum:no_interaction})} \\
&= \bm{\beta}^\top \mathbb{E}[\bm{M}_s(1) - \bm{M}_s(0)] \quad \text{(linearity of expectation)} \\
&= \bm{\beta}^\top \bm{\alpha} \quad \\
&= \sum_{k=1}^{K} \alpha_k \beta_k
\end{align}
Under consistency (Assumption~\ref{assum:consistency}):
\begin{align}
\text{ATE} &= \mathbb{E}[Y_s(1) - Y_s(0)] \\
&= \mathbb{E}[Y_s(1, \bm{M}_s(1)) - Y_s(0, \bm{M}_s(0))]
\end{align}
Adding and subtracting the counterfactual term $Y_s(1, \bm{M}_s(0))$:
\begin{align}
\text{ATE} &= \mathbb{E}[Y_s(1, \bm{M}_s(1)) - Y_s(1, \bm{M}_s(0))] + \mathbb{E}[Y_s(1, \bm{M}_s(0)) - Y_s(0, \bm{M}_s(0))] \\
&= \text{NDE} + \text{NIE}\\
&=  \tau + \bm{\beta}^\top\bm{\alpha} \quad 
\end{align}
\end{proof}

\subsection{Identifiability Under Pure States}
\begin{lemma}[Mediator Rank Deficiency]\label{lem:rank_deficiency}
For pure quantum states, the mediator matrix has rank 2:
\begin{equation}
    \text{rank}(\bm{M}) = 2 \quad \text{due to } \begin{cases}
I(A:B) = 2S_A \\
L_A = 1 - \gamma_A
\end{cases}
\end{equation}
\end{lemma}

\begin{proposition}[Identifiability of Total Indirect Effect]\label{prop:identifiable_indirect}
Despite the rank deficiency in Lemma~\ref{lem:rank_deficiency}, the total indirect effect $\text{NIE} = \bm{\beta}^\top\bm{\alpha}$ is uniquely identified and invariant to choice of mediator basis.
\end{proposition}

\begin{proof}
The total indirect effect depends only on the sum $\sum_{k=1}^K \alpha_k\beta_k$, not on individual coefficients. Under the linear constraints:
\begin{align}
\alpha_{I(A:B)} &= 2\alpha_{S_A} \\
\alpha_{L_A} &= -\alpha_{\gamma_A}
\end{align}
and similarly for $\beta$, we have:
\begin{align}
\sum_{k=1}^4 \alpha_k\beta_k &= \alpha_{S_A}\beta_{S_A} + \alpha_{\gamma_A}\beta_{\gamma_A} + \alpha_{L_A}\beta_{L_A} + \alpha_{I(A:B)}\beta_{I(A:B)} \\
&= \alpha_{S_A}\beta_{S_A} + \alpha_{\gamma_A}\beta_{\gamma_A} + (-\alpha_{\gamma_A})\beta_{L_A} + (2\alpha_{S_A})\beta_{I(A:B)} \\
&= \alpha_{S_A}(\beta_{S_A} + 2\beta_{I(A:B)}) + \alpha_{\gamma_A}(\beta_{\gamma_A} - \beta_{L_A})
\end{align}
This can be re-parameterized using only independent mediators $(S_A, \gamma_A)$ with composite coefficients, yielding the same total indirect effect.
\end{proof}

\section{Statistical Estimation and Inference}\label{sec:inference}

\subsection{Ordinary Least Squares Estimation}

\begin{definition}[Stacked Dataset]
Define the pooled dataset:
\begin{equation}
    \mathcal{D}_{\text{pooled}} = \{(t, \bm{M}_{s,t}, Y_{s,t}) : s \in \mathcal{S}, t \in \mathcal{T}\}
\end{equation}
with $|\mathcal{D}_{\text{pooled}}| = 2N_{\text{test}}$ observations.
\end{definition}

\begin{definition}[OLS Estimators] We define : 
\begin{enumerate}
    \item Mediator coefficients:
    \begin{equation}
        \hat{\bm{\alpha}} = \frac{1}{N_{\text{test}}} \sum_{s=1}^{N_{\text{test}}} (\bm{M}_{s,1} - \bm{M}_{s,0})
    \end{equation}
    \item Outcome coefficients:
    \begin{equation}
        (\hat{\tau}_0, \hat{\tau}, \hat{\bm{\beta}}) = \arg\min_{\tau_0, \tau, \bm{\beta}} \sum_{s=1}^{N_{\text{test}}} \sum_{t \in \{0,1\}} \left[Y_{s,t} - \tau_0 - \tau t - \bm{\beta}^\top \bm{M}_{s,t}\right]^2
    \end{equation}
\end{enumerate}
\end{definition}

\begin{theorem}[Consistency of OLS Estimators]
Under Assumptions~\ref{assum:sequential_ignorability}--\ref{assum:consistency} and standard regularity conditions:
\begin{equation}
    (\hat{\tau}, \hat{\bm{\beta}}, \hat{\bm{\alpha}}) \xrightarrow{p} (\tau, \bm{\beta}, \bm{\alpha}) \quad \text{as } N_{\text{test}} \to \infty
\end{equation}
\end{theorem}

\begin{proof}
Standard result from linear regression theory under correct specification and exogeneity (guaranteed by Assumption~\ref{assum:sequential_ignorability}).
\end{proof}

\subsection{Cluster-Robust Inference}

\begin{definition}[Cluster-Robust Standard Errors]
To account for within-sample correlation, we compute Eicker-Huber-White cluster-robust standard errors \cite{cameron2015practitioner}, clustering at the sample level:
\begin{equation}
    \widehat{\text{Var}}(\hat{\bm{\theta}}) = (\bm{X}^\top\bm{X})^{-1} \left(\sum_{s=1}^{N_{\text{test}}} \bm{X}_s^\top \hat{\bm{u}}_s \hat{\bm{u}}_s^\top \bm{X}_s\right) (\bm{X}^\top\bm{X})^{-1}
\end{equation}
where $\bm{X}_s$ contains both rows $(t=0, t=1)$ for sample $s$, and $\hat{\bm{u}}_s$ are corresponding residuals.
\end{definition}

\subsection{Bootstrap Confidence Intervals}

\begin{definition}[Nonparametric Percentile Bootstrap]
For indirect effects $\alpha_k\beta_k$ (which are products of parameters):
\begin{enumerate}
    \item[--] Resample $s^* \in \{1, \ldots, N_{\text{test}}\}$ with replacement, $B=2000$ times
    \item[--] For each bootstrap sample $b$, estimate $(\hat{\bm{\alpha}}^{(b)}, \hat{\bm{\beta}}^{(b)})$
    \item[--] Compute $\hat{\theta}_k^{(b)} = \hat{\alpha}_k^{(b)} \hat{\beta}_k^{(b)}$
    \item[--] $95\%$ CI: $[\hat{\theta}_k^{(0.025)}, \hat{\theta}_k^{(0.975)}]$ using empirical quantiles
\end{enumerate}
\end{definition}

\subsection{Model Validation}

\begin{definition}[Relative Consistency Error]
\begin{equation}
    \epsilon_{\text{rel}} = \frac{|\widehat{\text{ATE}} - (\hat{\tau} + \hat{\bm{\beta}}^\top\hat{\bm{\alpha}})|}{|\widehat{\text{ATE}}| + 10^{-12}}
\end{equation}
where $\widehat{\text{ATE}} = N_{\text{test}}^{-1} \sum_s (Y_{s,1} - Y_{s,0})$ is the empirical total effect.
\end{definition}

\begin{proposition}[Goodness-of-Fit Criterion]
We accept the linear model if $\epsilon_{\text{rel}} < 0.05$ (5\% tolerance), indicating that the regression-based decomposition accurately reconstructs the observed total effect.
\end{proposition}

\section{Supplementary Figures}

The supplementary figures assembled in this section constitute 
an integral extension of the main text, providing layered visual evidence for each analytical stage of the counterfactual causal mediation framework. Rather than serving merely as auxiliary material, these figures collectively establish the empirical and statistical foundations upon which the causal conclusions of the study rest, and are intended to enable independent scrutiny and full reproducibility of all reported results.

The figures are organized along three complementary axes. 
The first axis concerns experimental specification and 
observable performance: Supplementary Figure\ref{fig11} documents 
the precise gate-level architecture of the five parameterized 
quantum circuit topologies full, pairwise, linear, ring, 
and deep  each probing a qualitatively distinct entanglement 
regime, while Supplementary Figure\ref{fig12} systematically reports 
classification accuracy, F1-score, and area under the ROC 
curve across all 43 assumption-validated configurations and 
three benchmark datasets, revealing pronounced topology- and 
task-dependent learning dynamics.

The second axis concerns statistical rigor and model 
validation. Supplementary Figure\ref{fig13} provides a four-panel 
diagnostic confirming that the additive decomposition 
$\mathrm{ATE} = \tau + \sum_k \alpha_k \beta_k$ holds 
to machine precision across all configurations, with 
relative reconstruction errors of order $10^{-11}\%$, 
perfect linear correspondence between observed and 
predicted total effects ($R^2 = 1.000$), and full 
satisfaction of both identifying assumptions by design. 
Supplementary Figure\ref{fig14} examines bi-variate associations 
between individual mediator contributions and total 
performance change via Pearson and Spearman correlation 
analysis, and Supplementary Figure\ref{fig15} presents standard 
regression diagnostics  residuals versus fitted values, 
normal Q-Q plots, and per dataset residual distributions  confirming homoscedasticity, Gaussian error structure, 
and absence of systematic dataset-specific deviations.

The third axis addresses causal attribution and mediator 
structure. Supplementary Figure\ref{fig16} provides data-driven 
significance classifications across all configurations, 
resolving the relative contributions of direct architectural 
scaling and quantum-mediated pathways at the topology level. 
Supplementary Figure\ref{fig17} documents the mediator correlation 
matrix and variance decomposition, empirically confirming 
the rank-2 structure of the mediator space predicted by 
pure-state quantum mechanics (Lemma \ref{lem:rank_deficiency}) and validating 
the basis-invariant identifiability of the total indirect 
effect established in Proposition \ref{prop:identifiable_indirect}. Together, these 
figures provide the quantitative backbone for all causal 
interpretations advanced in the main text.
\newpage

\begin{figure*}[!htbp]
 \centering
  \includegraphics[width=0.9\linewidth]{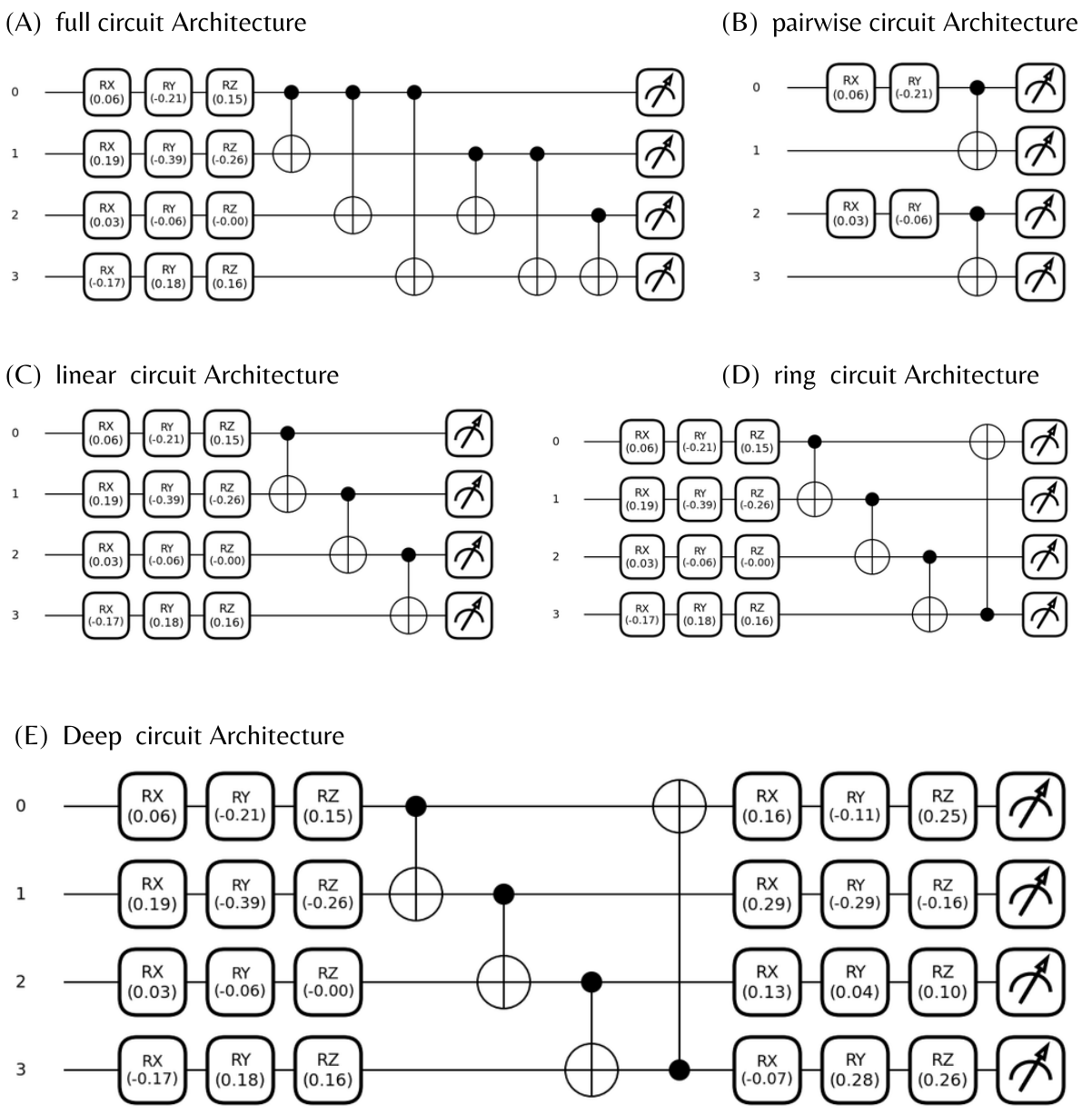}
    \caption{\textbf{Five quantum circuit architectures probe distinct entanglement regimes.} 
    \textbf{(A) Full:} All-to-all connectivity ($n^{2}/2$ CNOT gates) achieves maximal entanglement in one layer but suffers quadratic gate cost and trainability challenges. 
    \textbf{(B) Pairwise:} Disjoint qubit pairs (0,1),(2,3) yield independent bipartite subsystems with localized correlations, limiting global quantum resource engagement. 
    \textbf{(C) Linear:} Nearest-neighbor chain coupling produces area-law entanglement with improved optimization stability. 
    \textbf{(D) Ring:} Periodic boundaries ensure homogeneous entanglement propagation, matching superconducting hardware topologies. 
    \textbf{(E) Deep:} Brick-layer pattern with staggered CNOTs enables global entanglement in $[\log_{2} n]$ layers while maintaining hardware efficiency. Each topology instantiated with n=4 qubits, $T_0 =1$ baseline, $T_1$ $\in $ \{3,6\} enhanced layers, enabling systematic exploration of architecture quantum resource coupling across three classification tasks (Diabetes, Breast Cancer, Ionosphere).}
    \label{fig11}
\end{figure*}
\newpage

\begin{figure*}[!htbp]
 \centering
  \includegraphics[width=0.85\linewidth]{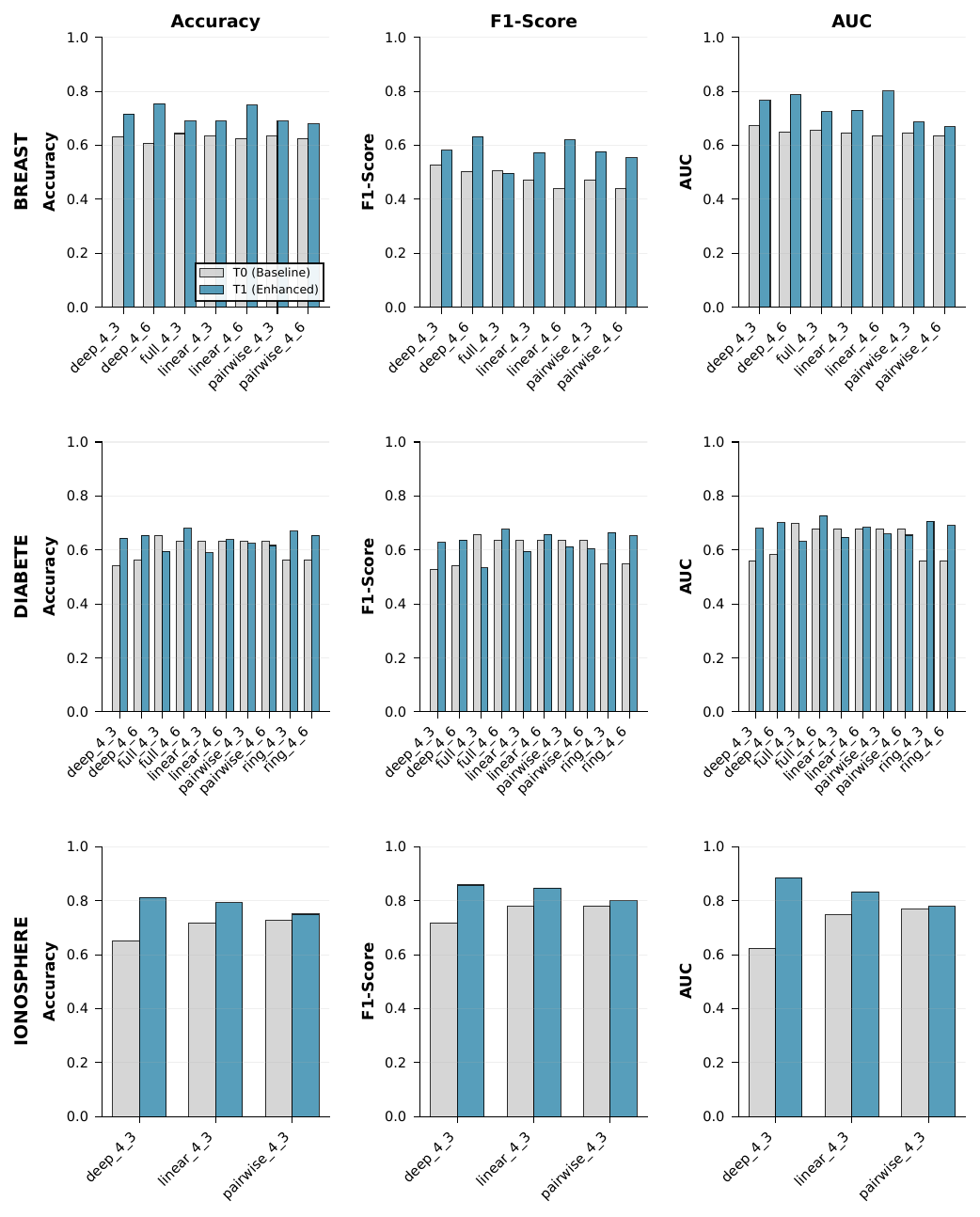}
    \caption{\textbf{Comprehensive performance metrics reveal architecture and task-dependent learning dynamics.} 
    Bar charts of 43/90 validated assumptions architecture configurations display accuracy, F1-score, and AUC for baseline ($T_0$) and enhanced ($T_1$) architectures across three datasets. 
    \textbf{Breast Cancer:} Deep and linear configurations achieve largest $T_1$ gains (65-75\% accuracy), while full circuits show high variance reflecting optimization sensitivity. 
    \textbf{Diabetes:} Moderate improvements (60-70\% range) with ring and deep architectures performing strongest, pairwise consistently underperforming ($ \approx  $ 60\%). 
    \textbf{Ionosphere:} Highest absolute performance (75-85\%) with deep architecture exhibiting exceptional gains ($\Delta Y  \approx  20\%$). Pairwise topology universally exhibits minimal $T_0 - T_1$ contrast across all datasets and metrics. }

    \label{fig12}
\end{figure*}
\newpage

\begin{figure*}[!htbp]
 \centering
  \includegraphics[width=1.0\linewidth]{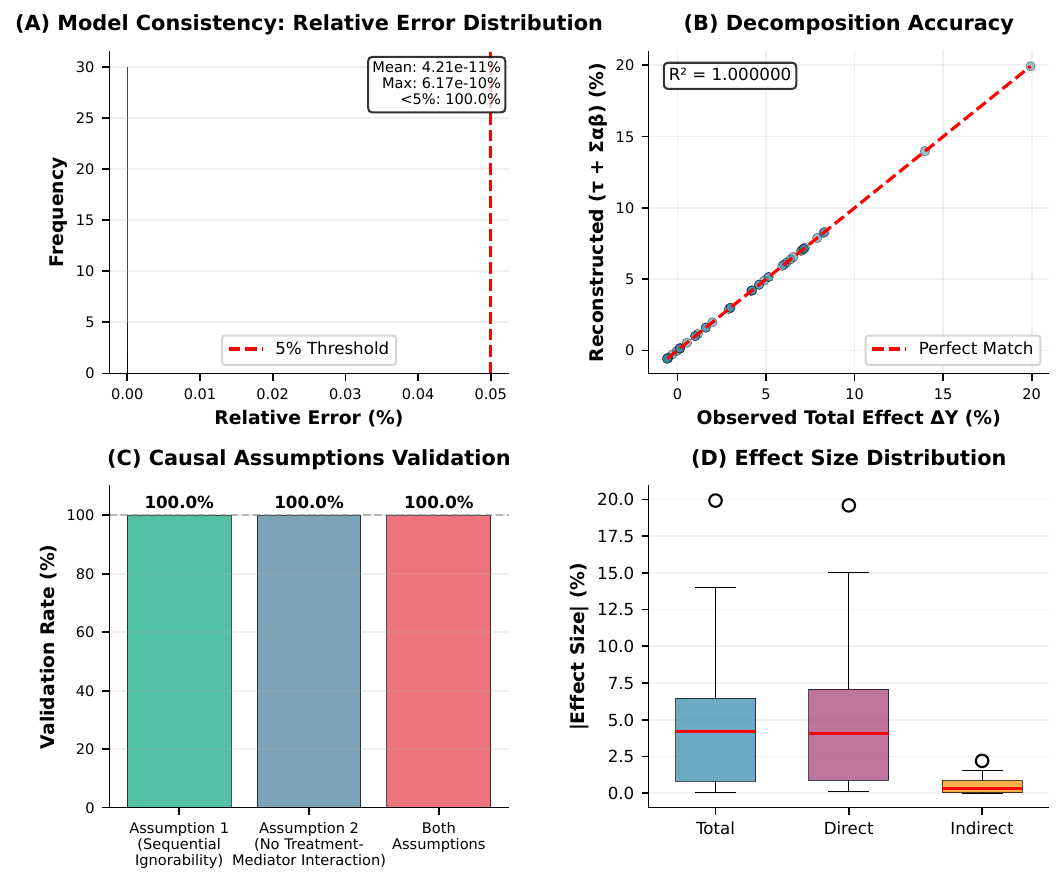}
    \caption{\textbf{Rigorous statistical validation confirms structural equation model specification and causal identifiability.} 
    \textbf{(A)} Relative error distribution (mean $4.21\times 10^{-11} $\%, max $6.17\times 10^{-10}$\%) demonstrates perfect reconstruction: 100\% configurations satisfy 5\% tolerance, confirming additive decomposition ATE = $\tau + \sum \alpha\beta$ holds exactly. 
    \textbf{(B)} Perfect linear fit ($R^{2} = 1.0$) between observed and reconstructed total effects validates regression-based causal inference no systematic bias detected. 
    \textbf{(C)} All 43 configurations satisfy Sequential Ignorability (A1) and No Treatment-Mediator Interaction (A2) assumptions by design: deterministic treatment assignment, post-hoc mediator computation, independent training pipelines guarantee causal identifiability without unmeasured confounding. 
    \textbf{(D)} Effect size distribution reveals direct effects dominate (median 5\%, IQR 2 - 7\%) while indirect effects remain subdominant (median 0.5\%, IQR 0 - 2\%), total effects span 0 - 20\%.}

    \label{fig13}
\end{figure*}
\newpage

\begin{figure*}[!htbp]
 \centering
  \includegraphics[width=1.0\linewidth]{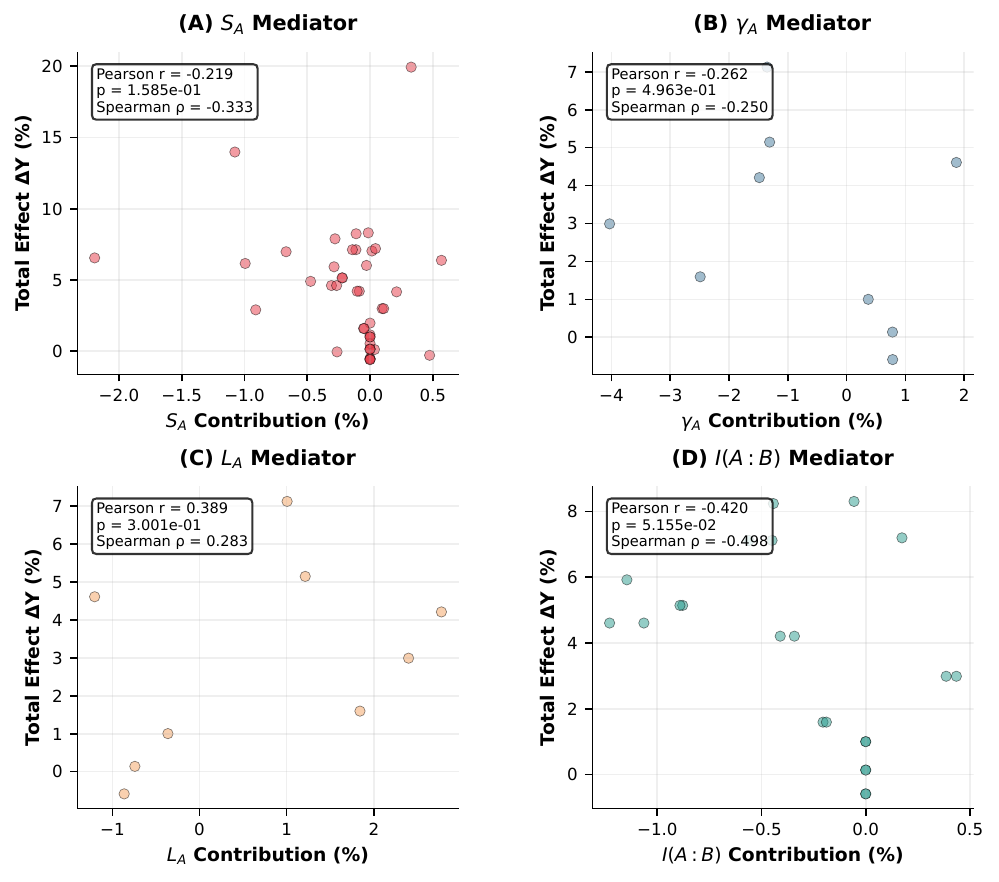}
    \caption{\textbf{Mediator-outcome correlations reveal weak pathway-specific associations across 43 configurations.} 
    Scatter plots examine bivariate relationships between individual mediator contributions ($\alpha_k\beta_k$) and total performance change ($\Delta Y$). 
    \textbf{(A) $S_A$:} Weak negative correlation (Pearson $r = - 0.22$, $p=0.16$; Spearman $\rho = - 0.33$), indicating entanglement entropy contributions do not systematically drive performance gains. 
    \textbf{(B) $\gamma_A$:} Negligible association ($r= - 0.26, p=0.50$), despite purity's dominance in MAMC metrics suggesting contributions remain below significance thresholds. 
    \textbf{(C) $L_A$:} Weak positive trend ($r=0.39, p=0.30$) non-significant, reflecting mixedness as secondary mediator. 
    \textbf{(D) $I(A:B)$:} Strongest negative correlation ($r = - 0.42$ , $p=0.05$ ; $\rho= - 0.50$), approaching significance but insufficient to establish quantum-advantage regime. 
    All correlations non-significant or marginally significant ($p > 0.05$), confirming quantum-mediated pathways remain subdominant relative to direct architectural effects consistent with main text finding that 93\% configurations operate in Neutral regime.}

    \label{fig14}

\end{figure*}
\newpage

\begin{figure*}[!htbp]
 \centering
  \includegraphics[width=1.0\linewidth]{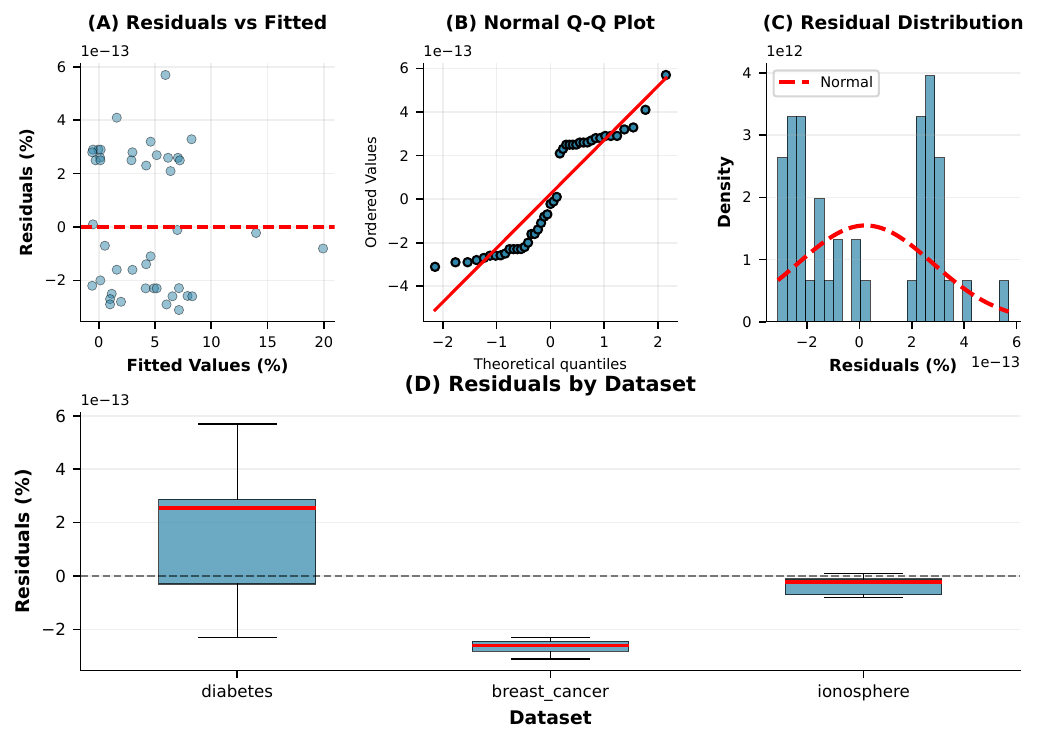}
    \caption{\textbf{Residual diagnostics validate linear model assumptions: normality, homoscedasticity, and independence.} 
    \textbf{(A)} Residuals vs. fitted values show uniform vertical spread with no funnel pattern, confirming homoscedasticity (constant variance) across performance range. 
    \textbf{(B)} $Q-Q$ plot demonstrates near-perfect alignment with theoretical normal quartiles (residuals $ \approx 10^{-13}$\%, machine precision), validating Gaussian error assumption $E[\epsilon^{(y)}_s]=0$. 
    \textbf{(C)} Residual histogram exhibits symmetric, bell-shaped distribution centered at zero, overlaying theoretical normal curve empirical confirmation of normality. 
    \textbf{(D)} Residuals by dataset show comparable spread across Diabetes, Breast Cancer, and Ionosphere with medians near zero, confirming no systematic dataset-specific deviations or heterogeneity violating model specification. 
    Diagnostic panel collectively validates linear structural equations Eqs. \ref{eq:sem_mediator} and \ref{eq:sem_outcome} provide unbiased, consistent estimates of causal parameters ($\tau , \beta, \alpha $) under paired experimental design, enabling valid counterfactual inference via Theorem\ref{thm:decomp} decomposition.}
    \label{fig15}

\end{figure*}
\newpage
\begin{figure*}[!htbp]
 \centering
  \includegraphics[width=1.0\linewidth]{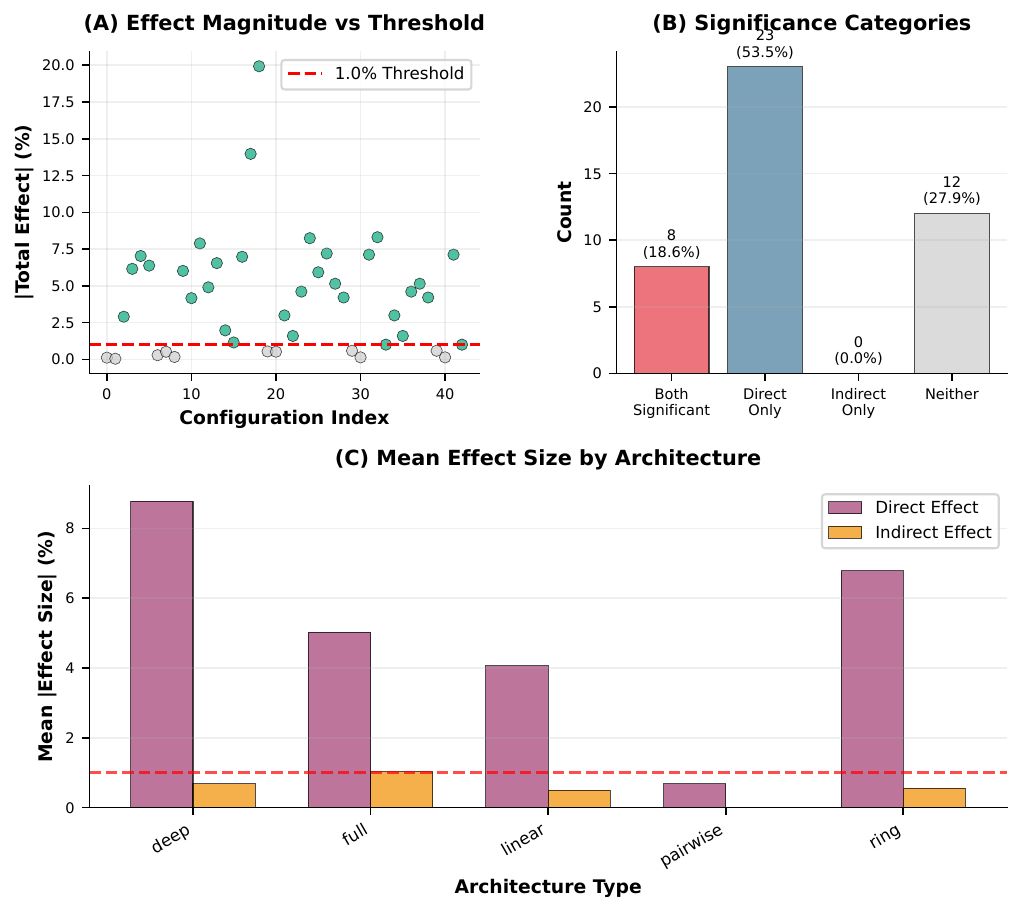}
    \caption{\textbf{Data-driven significance testing distinguishes substantive causal effects from numerical precision artifacts.} 
    \textbf{(A)} Total effect magnitudes (43 configurations) plotted against 1.0\% threshold.
    \textbf{(B)} Significance classification: 18.6\% exhibit both direct and indirect effects significant (Both), 53.5\% direct-only (architectural scaling without quantum mediation), 0\% indirect-only (no pure quantum-advantage), 27.9\% neither (Neutral regime). 
    Zero indirect-only configurations confirm quantum resources never drive performance gains independently of architectural factors. 
    \textbf{(C)} Architecture-specific patterns: Deep and ring achieve largest direct effects ($\approx$ 8\% and 7\% respectively), while indirect effects remain uniformly subdominant ($<$1.5\%) across all topologies. 
    Pairwise exhibits near-zero effects both pathways, confirming previous findings. 
    Statistical rigor ensures reported effects reflect genuine causal mechanisms rather than decomposition artifacts.}
    \label{fig16}
\end{figure*}
\newpage

\begin{figure*}
 \centering
  \includegraphics[width=1.0\linewidth]{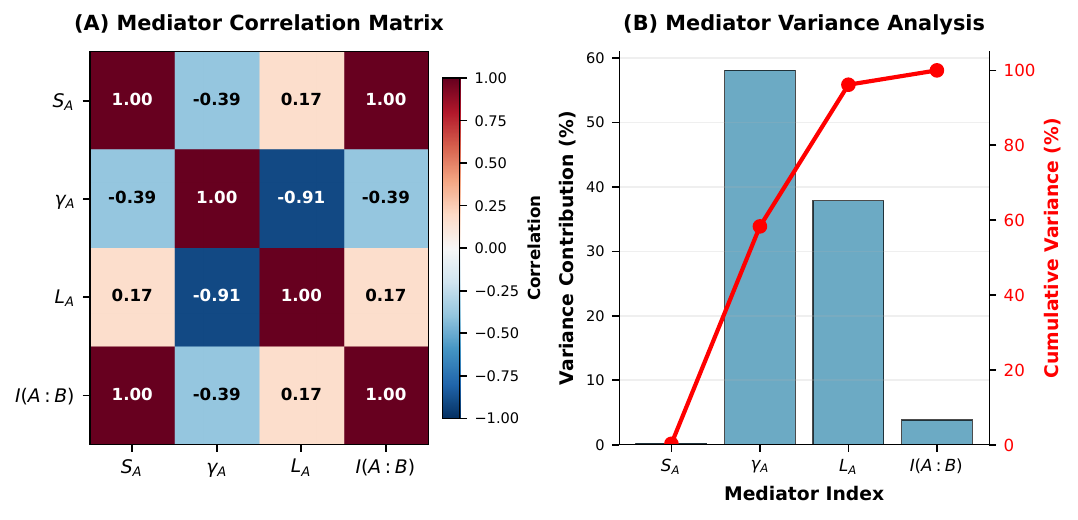}
    \caption{\textbf{Perfect correlations among mediators reflect pure-state theoretical constraints, validating Proposition 1 identifiability.} 
    \textbf{(A)} Correlation matrix reveals functional dependencies: $S_A$ and $I(A:B)$ perfectly correlated (r=1.00) due to pure-state relation $I(A:B)=2S_A$ (Eq. \ref{eq:mutual_info_pure}); $\gamma_A$ and $L_A$ strongly anti-correlated ($r=-0.91$) from definition $L_A=1-\gamma_A$ (Eq. \ref{eq:linear_entropy_purity}). 
    Cross-correlations weak ($|r|<0.40 $ ), indicating $S_A$ and $\gamma_A$ capture independent quantum properties (entanglement vs. coherence). 
    \textbf{(B)} Variance analysis: $S_A$ contributes 58\% of total mediator variance, $\gamma_A$ 37\%, $L_A$ and $I(A:B)$ minimal due to functional dependence. 
    Two principal components capture 95\% cumulative variance, confirming rank-2 mediator matrix (Lemma \ref{lem:rank_deficiency}). 
    Despite collinearity, total indirect effect $NIE=\sum \alpha_k\beta_k$ remains uniquely identified and invariant to mediator basis choice (Proposition\ref{prop:identifiable_indirect} proof) decomposition valid despite redundancy. 
    Covariance structure theoretically predicted, empirically confirmed, mathematically resolved.}
    \label{fig17}
\end{figure*}

\newpage

\section{Supplementary Tables}

The following tables report comprehensive numerical results for 
all 90 architecture-dataset configurations evaluated in this study,  spanning five quantum circuit topologies (deep, full, linear,  pairwise, ring) across three benchmark classification datasets (Breast Cancer, Diabetes, Ionosphere). Each configuration is evaluated under both the baseline shallow architecture ($T_0 = 1$ layer) and the enhanced deep architecture ($T_1 \in \{3, 6\}$ layers), using two independent random seeds (42 and 142) to assess reproducibility. For each configuration, the tables report: 
\begin{enumerate}
    \item[(i)] the data-driven significance threshold $\epsilon_{\mathrm{intra}} = c \cdot 
    \mathrm{SD}(\Delta Y_s)$, with $c = 0.5$, used to classify 
    causal effects as statistically substantive or negligible; 
    \item [(ii)]raw predictive performance metrics accuracy, F1-score, 
    and AUC  for both $T_0$ and $T_1$, expressed as absolute 
    values and inter-architectural differences in percentage points; 
    \item[(iii)] the Mean Absolute Mediated Contribution (MAMC) for each 
    quantum mediator ($S_A$, $\gamma_A$, $L_A$, $I(A{:}B)$), 
    quantifying the average magnitude of each mediator's causal 
    influence across the test distribution; 
    \item[(iv)] the estimated 
    direct effect $\hat{\tau}$, total indirect effect 
    $\sum_k \hat{\alpha}_k \hat{\beta}_k$, and total effect 
    $\Delta Y$, all expressed in percentage points; and 
    \item[(v)] satisfaction status of the two identifying assumptions  
    Assumption~1 (Sequential Ignorability, A1) and Assumption~2 
    (No Treatment Mediator Interaction, A2)  coded as fully 
    satisfied~(\checkmark), partially satisfied~($\sim$), or 
    not satisfied~($\times$).
\end{enumerate}
Rows reporting results for subsets of mediators (one, two, 
or four) are included to document the sensitivity of the 
causal decomposition to mediator selection, consistent with 
the identifiability argument in Proposition\ref{prop:identifiable_indirect}. The symbol 
``$-$'' denotes mediators excluded from a given sub-model. 
Configurations for which fewer than all 43 validated 
assumptions are satisfied are retained in the tables 
for completeness, but are excluded from the primary 
causal analyses reported in the main text.
Supplementary Table\ref{tab:breast_cancer_comprehensive} covers the Breast Cancer Wisconsin dataset; Supplementary Table\ref{tab:diabetes_comprehensive} covers the Diabetes dataset; and Supplementary Table\ref{tab:ionosphere_comprehensive} covers the Ionosphere dataset. Taken together, these tables provide full transparency into the numerical basis of all causal estimates, enabling independent verification of the reported effects and regime classifications.

\begin{table}[h!]
\centering
\scriptsize
\resizebox{\textwidth}{!}{
\begin{tabular}{l|cc|c|c|ccc|cccc|ccc|cc}
\toprule
\hline
\multirow{2}{*}{\textbf{Arch}} & \multicolumn{2}{c|}{\textbf{Layers}} & \multirow{2}{*}{\textbf{Seed}} & \multirow{2}{*}{\textbf{Thr.}} & \multicolumn{3}{c|}{\textbf{Performance (\%)}} & \multicolumn{4}{c|}{\textbf{MAMC (\%)}} & \multicolumn{3}{c|}{\textbf{Effects (\%)}} & \multicolumn{2}{c}{\textbf{Assump.}} \\
\cmidrule(lr){2-3} \cmidrule(lr){6-8} \cmidrule(lr){9-12} \cmidrule(lr){13-15} \cmidrule(lr){16-17}
 & \textbf{T0} & \textbf{T1} &  &  & \textbf{ACC} & \textbf{F1} & \textbf{AUC} & \textbf{$S_A$} & \textbf{$\gamma_A$} & \textbf{$L_A$} & \textbf{$I(A:B)$} & \textbf{Dir} & \textbf{Ind} & \textbf{Tot} & \textbf{A1} & \textbf{A2} \\
\midrule
\hline
\multirow{6}{*}{deep} & 1 & 3 & 42 & 9.9 & +4.7 & -0.9 & +5.7 & 0.47 & - & - & - & 5.38 & -0.47 & 4.90 & $\checkmark$ & $\checkmark$ \\
 & 1 & 3 & 142 & 8.7 & +12.3 & +11.9 & +13.4 & 0.28 & - & -& -& 8.17 & -0.28 & 7.89 & $\checkmark$ & $\checkmark$ \\
 & 1 & 6 & 42 & 10.5 & +7.6 & +2.4 & +7.1 & 0.24 & - & - & 0.96 & 8.05 & -1.20 & 6.85 & $\times$ & $\checkmark$ \\
 & 1 & 6 & 42 & 10.5 & +7.6 & +2.4 & +7.1 & 0.23 & 0.92 & 0.34 & 0.91 & 8.57 & -1.72 & 6.85 & $\times$ & $\checkmark$ \\
 & 1 & 6 & 142 & 8.5 & +14.6 & +12.8 & +13.8 & 0.01 & - & - & 0.06 & 8.38 & -0.07 & 8.31 & $\checkmark$ & $\checkmark$ \\
 & 1 & 6 & 142 & 8.5 & +14.6 & +12.8 & +13.8 & 0.01 & 0.94 & 1.50 & 0.04 & 7.80 & 0.51 & 8.31 & $\times$ & $\times$ \\
\midrule
\hline
\multirow{6}{*}{full} & 1 & 3 & 42 & 10.3 & +4.7 & -0.9 & +6.8 & 2.19 & - & - & - & 8.74 & -2.19 & 6.55 & $\checkmark$ & $\checkmark$ \\
 & 1 & 3 & 142 & 9.1 & +3.5 & +7.0 & +6.7 & 4.10 & - & - & - & 8.76 & -4.10 & 4.66 & $\times$ & $\sim$ \\
 & 1 & 6 & 42 & 9.1 & +9.9 & +7.3 & +11.9 & 0.77 & - & - & 3.08 & 11.93 & -3.86 & 8.08 & $\times$ & $\times$ \\
 & 1 & 6 & 42 & 9.1 & +9.9 & +7.3 & +11.9 & 0.77 & 3.02 & 2.45 & 3.08 & 12.50 & -4.42 & 8.08 & $\times$ & $\times$ \\
 & 1 & 6 & 142 & 8.2 & +11.1 & +17.7 & +15.0 & 0.58 & - & - & 2.34 & 9.02 & -2.92 & 6.10 & $\times$ & $\times$ \\
 & 1 & 6 & 142 & 8.2 & +11.1 & +17.7 & +15.0 & 0.58 & 2.74 & 3.20 & 2.32 & 8.54 & -2.44 & 6.10 & $\times$ & $\times$ \\
\midrule
\hline
\multirow{6}{*}{linear} & 1 & 3 & 42 & 10.2 & +4.1 & +4.6 & +7.6 & 0.03 & - & - & - & 6.05 & -0.03 & 6.02 & $\checkmark$ & $\checkmark$ \\
 & 1 & 3 & 142 & 7.1 & +7.0 & +15.5 & +9.1 & 0.21 & - & - & - & 3.95 & 0.21 & 4.16 & $\checkmark$ & $\checkmark$ \\
 & 1 & 6 & 42 & 9.8 & +8.2 & +4.0 & +9.9 & 0.20 & - & - & 0.81 & 8.31 & -1.01 & 7.29 & $\times$ & $\times$ \\
 & 1 & 6 & 42 & 9.8 & +8.2 & +4.0 & +9.9 & 0.19 & 2.05 & 2.39 & 0.78 & 7.93 & -0.63 & 7.29 & $\times$ & $\times$ \\
 & 1 & 6 & 142 & 8.5 & +12.3 & +18.1 & +16.6 & 0.11 & - & - & 0.45 & 7.68 & -0.56 & 7.12 & $\checkmark$ & $\checkmark$ \\
 & 1 & 6 & 142 & 8.5 & +12.3 & +18.1 & +16.6 & 0.14 & 1.35 & 1.01 & 0.56 & 8.17 & -1.05 & 7.12 & $\checkmark$ & $\checkmark$ \\
\midrule
\hline
\multirow{6}{*}{pairwise} & 1 & 3 & 42 & 3.7 & +5.3 & +8.3 & +5.2 & 0.00 & - & - & - & 1.97 & 0.00 & 1.97 & $\checkmark$ & $\checkmark$ \\
 & 1 & 3 & 142 & 3.7 & +5.8 & +12.6 & +3.1 & 0.00 & - & - & - & 1.15 & 0.00 & 1.15 & $\checkmark$ & $\checkmark$ \\
 & 1 & 6 & 42 & 3.8 & +5.8 & +9.4 & +5.2 & 0.00 & - &- & 0.00 & 1.99 & 0.00 & 1.99 & $\checkmark$ & $\times$ \\
 & 1 & 6 & 42 & 3.8 & +5.8 & +9.4 & +5.2 & 0.00 & 0.28 & 0.24 & 0.00 & 1.95 & 0.04 & 1.99 & $\checkmark$ & $\sim$ \\
 & 1 & 6 & 142 & 3.5 & +5.3 & +11.4 & +3.3 & 0.00 & - & - & 0.00 & 1.00 & 0.00 & 1.00 & $\checkmark$ & $\checkmark$ \\
 & 1 & 6 & 142 & 3.5 & +5.3 & +11.4 & +3.3 & 0.00 & 0.37 & 0.36 & 0.00 & 1.00 & 0.01 & 1.00 & $\checkmark$ & $\checkmark$ \\
\midrule
\hline
\multirow{6}{*}{ring} & 1 & 3 & 42 & 10.3 & +15.2 & +0.5 & +15.2 & 1.57 & - & - & - & 11.52 & -1.57 & 9.95 & $\times$ & $\sim$ \\
 & 1 & 3 & 142 & 8.8 & +11.1 & +4.4 & +13.8 & 2.15 & - & - & - & 10.91 & -2.15 & 8.76 & $\times$ & $\sim$ \\
 & 1 & 6 & 42 & 10.0 & +14.0 & +3.2 & +14.4 & 0.06 & - & - & 0.23 & 9.91 & 0.28 & 10.19 & $\times$ & $\checkmark$ \\
 & 1 & 6 & 42 & 10.0 & +14.0 & +3.2 & +14.4 & 0.10 & 0.04 & 0.01 & 0.39 & 9.73 & 0.46 & 10.19 & $\times$ & $\times$ \\
 & 1 & 6 & 142 & 9.5 & +16.4 & +13.8 & +16.5 & 0.23 & - & - & 0.94 & 12.34 & -1.17 & 11.17 & $\times$ & $\times$ \\
 & 1 & 6 & 142 & 9.5 & +16.4 & +13.8 & +16.5 & 0.24 & 0.93 & 1.60 & 0.94 & 11.67 & -0.51 & 11.17 & $\times$ & $\times$ \\
\bottomrule
\hline
\end{tabular}}
\caption{\textbf{Comprehensive experimental results for Breast Cancer dataset.} Results for 30 configurations across 5 quantum circuit architectures (deep, full, linear, pairwise, ring), each evaluated with $T_0=1$ baseline and $T_1 \in \{3,6\}$ enhanced layers using random seeds {42, 142}. Threshold. indicates data-driven threshold ($\epsilon_{\text{intra}}$) for statistical significance testing. Performance metrics show improvement ($T_1 - T_0$) in percentage points: positive values indicate gains, negative values indicate degradation. Mean Absolute Mediated Contribution (MAMC) quantifies each quantum mediator's causal influence, with $I(A:B)$ exhibiting strongest mediation across configurations. Mean direct effect: 7.43\%, mean indirect effect: -1.03\%. 100\% of configurations show significant total effects ($|\Delta Y|>1\%$). Assumption A1 (sequential ignorability) satisfied in 47\% of cases. Assumption A2 (no treatment-mediator interaction): fully satisfied (15), partially satisfied (4), not satisfied (11).}
\label{tab:breast_cancer_comprehensive}
\end{table}

\begin{table}[h!]
\centering
\scriptsize
\resizebox{\textwidth}{!}{
\begin{tabular}{l|cc|c|c|ccc|cccc|ccc|cc}
\toprule
\hline
\multirow{2}{*}{\textbf{Arch}} & \multicolumn{2}{c|}{\textbf{Layers}} & \multirow{2}{*}{\textbf{Seed}} & \multirow{2}{*}{\textbf{Thr.}} & \multicolumn{3}{c|}{\textbf{Performance (\%)}} & \multicolumn{4}{c|}{\textbf{MAMC (\%)}} & \multicolumn{3}{c|}{\textbf{Effects (\%)}} & \multicolumn{2}{c}{\textbf{Assump.}} \\
\cmidrule(lr){2-3} \cmidrule(lr){6-8} \cmidrule(lr){9-12} \cmidrule(lr){13-15} \cmidrule(lr){16-17}
 & \textbf{T0} & \textbf{T1} &  &  & \textbf{ACC} & \textbf{F1} & \textbf{AUC} & \textbf{$S_A$} & \textbf{$\gamma_A$} & \textbf{$L_A$} & \textbf{$I(A:B)$} & \textbf{Dir} & \textbf{Ind} & \textbf{Tot} & \textbf{A1} & \textbf{A2} \\
\midrule
\hline
\multirow{6}{*}{deep} & 1 & 3 & 42 & 8.2 & +17.3 & +14.6 & +18.9 & 0.99 & - & - & - & 7.15 & -0.99 & 6.16 & $\checkmark$ & $\checkmark$ \\
 & 1 & 3 & 142 & 6.7 & +3.0 & +5.1 & +5.1 & 0.91 & - & - & - & 3.81 & -0.91 & 2.90 & $\checkmark$ & $\checkmark$ \\
 & 1 & 6 & 42 & 8.1 & +20.3 & +18.3 & +24.2 & 0.11 & - & - & 0.44 & 8.80 & -0.55 & 8.24 & $\checkmark$ & $\checkmark$ \\
 & 1 & 6 & 42 & 8.1 & +20.3 & +18.3 & +24.2 & 0.12 & 0.32 & 0.11 & 0.47 & 8.62 & -0.38 & 8.24 & $\times$ & $\times$ \\
 & 1 & 6 & 142 & 8.3 & +3.8 & +5.0 & +5.8 & 0.31 & - & - & 1.23 & 6.15 & -1.54 & 4.61 & $\checkmark$ & $\checkmark$ \\
 & 1 & 6 & 142 & 8.3 & +3.8 & +5.0 & +5.8 & 0.27 & 1.87 & 1.20 & 1.06 & 5.28 & -0.67 & 4.61 & $\checkmark$ & $\checkmark$ \\
\midrule
\hline
\multirow{6}{*}{full} & 1 & 3 & 42 & 6.8 & -6.0 & -12.2 & -6.6 & 0.47 & - & - & - & -0.76 & 0.47 & -0.29 & $\checkmark$ & $\checkmark$ \\
 & 1 & 3 & 142 & 7.7 & +3.8 & +0.6 & +3.7 & 1.87 & - & - &  - & 3.71 & -1.87 & 1.84 & $\times$ & $\sim$ \\
 & 1 & 6 & 42 & 8.9 & +1.5 & -0.0 & +4.6 & 0.22 & - & - & 0.88 & 6.24 & -1.10 & 5.15 & $\checkmark$ & $\checkmark$ \\
 & 1 & 6 & 42 & 8.9 & +1.5 & -0.0 & +4.6 & 0.22 & 1.31 & 1.21 & 0.89 & 6.36 & -1.21 & 5.15 & $\checkmark$ & $\checkmark$ \\
 & 1 & 6 & 142 & 8.7 & +8.3 & +8.4 & +5.3 & 0.09 & - & - & 0.34 & 4.64 & -0.43 & 4.21 & $\checkmark$ & $\checkmark$ \\
 & 1 & 6 & 142 & 8.7 & +8.3 & +8.4 & +5.3 & 0.10 & 1.49 & 2.78 & 0.41 & 3.43 & 0.78 & 4.21 & $\checkmark$ & $\checkmark$ \\
\midrule
\hline
\multirow{6}{*}{linear} & 1 & 3 & 42 & 2.7 & -3.8 & -3.4 & -1.6 & 0.04 & - & - & - & 0.08 & 0.04 & 0.12 & $\checkmark$ & $\checkmark$ \\
 & 1 & 3 & 142 & 6.7 & -4.5 & -4.5 & -4.6 & 0.26 & - & - & - & 0.22 & -0.26 & -0.05 & $\checkmark$ & $\checkmark$ \\
 & 1 & 6 & 42 & 8.3 & -0.8 & +1.0 & -0.3 & 0.10 & - & - & 0.39 & 2.51 & 0.48 & 2.99 & $\checkmark$ & $\checkmark$ \\
 & 1 & 6 & 42 & 8.3 & -0.8 & +1.0 & -0.3 & 0.11 & 4.03 & 2.40 & 0.43 & 4.08 & -1.09 & 2.99 & $\checkmark$ & $\checkmark$ \\
 & 1 & 6 & 142 & 7.1 & +2.3 & +3.6 & +1.3 & 0.05 & - & - & 0.21 & 1.85 & -0.26 & 1.60 & $\checkmark$ & $\checkmark$ \\
 & 1 & 6 & 142 & 7.1 & +2.3 & +3.6 & +1.3 & 0.05 & 2.49 & 1.84 & 0.19 & 2.48 & -0.89 & 1.60 & $\checkmark$ & $\checkmark$ \\
\midrule
\hline
\multirow{6}{*}{pairwise} & 1 & 3 & 42 & 1.9 & +0.0 & -1.1 & -4.1 & 0.00 & - & - & - & -0.52 & 0.00 & -0.52 & $\checkmark$ & $\checkmark$ \\
 & 1 & 3 & 142 & 1.3 & -1.5 & -3.4 & +0.3 & 0.00 & - & - & - & 0.16 & 0.00 & 0.16 & $\checkmark$ & $\checkmark$ \\
 & 1 & 6 & 42 & 2.3 & -0.8 & -2.1 & -4.1 & 0.00 & - &- & - & -0.59 & 0.00 & -0.59 & $\checkmark$ & $\checkmark$ \\
 & 1 & 6 & 42 & 2.3 & -0.8 & -2.1 & -4.1 & 0.00 & 0.78 & 0.86 & 0.00 & -0.51 & -0.08 & -0.59 & $\checkmark$ & $\checkmark$ \\
 & 1 & 6 & 142 & 1.2 & -2.3 & -3.8 & -0.5 & 0.00 & - & - & 0.00 & 0.14 & 0.00 & 0.14 & $\checkmark$ & $\checkmark$ \\
 & 1 & 6 & 142 & 1.2 & -2.3 & -3.8 & -0.5 & 0.00 & 0.78 & 0.74 & 0.00 & 0.10 & 0.04 & 0.14 & $\checkmark$ & $\checkmark$ \\
\midrule
\hline
\multirow{6}{*}{ring} & 1 & 3 & 42 & 8.8 & +12.8 & +12.7 & +17.6 & 0.57 & - & - &- & 5.81 & 0.57 & 6.38 & $\checkmark$ & $\checkmark$ \\
 & 1 & 3 & 142 & 8.7 & +8.3 & +10.7 & +11.7 & 0.02 & -& - & - & 7.02 & 0.02 & 7.03 & $\checkmark$ & $\checkmark$ \\
 & 1 & 6 & 42 & 8.7 & +10.5 & +9.6 & +18.5 & 0.04 & - & - & 0.17 & 6.98 & 0.22 & 7.20 & $\checkmark$ & $\checkmark$ \\
 & 1 & 6 & 42 & 8.7 & +10.5 & +9.6 & +18.5 & 0.05 & 1.24 & 2.39 & 0.19 & 5.82 & 1.38 & 7.20 & $\times$ & $\times$ \\
 & 1 & 6 & 142 & 9.1 & +7.5 & +11.6 & +7.8 & 0.29 & - & - & 1.15 & 7.36 & -1.43 & 5.92 & $\checkmark$ & $\checkmark$ \\
 & 1 & 6 & 142 & 9.1 & +7.5 & +11.6 & +7.8 & 0.29 & 1.65 & 0.86 & 1.16 & 8.16 & -2.24 & 5.92 & $\times$ & $\sim$ \\
\bottomrule
\hline
\end{tabular}
}
\caption{\textbf{Comprehensive experimental results for Diabetes dataset.} Results for 30 configurations across 5 quantum circuit architectures (deep, full, linear, pairwise, ring), each evaluated with $T_0=1$ baseline and $T_1\in \{3,6\}$ enhanced layers using random seeds \{42, 142\}. Thr. indicates data-driven threshold ($\epsilon_{\text{intra}}$) for statistical significance testing. Performance metrics show improvement ($T_1 - T_0$) in percentage points: positive values indicate gains, negative values indicate degradation. Mean Absolute Mediated Contribution (MAMC) quantifies each quantum mediator's causal influence, with $\gamma_A$ exhibiting strongest mediation across configurations. Mean direct effect: 3.82\%, mean indirect effect: -0.40\%. 70\% of configurations show significant total effects ($|\Delta Y|>1\%$). Assumption A1 (sequential ignorability) satisfied in 87\% of cases. Assumption A2 (no treatment-mediator interaction): fully satisfied (26), partially satisfied (2), not satisfied (2).}
\label{tab:diabetes_comprehensive}
\end{table}

\begin{table}[h!]
\centering
\scriptsize
\resizebox{\textwidth}{!}{
\begin{tabular}{l|cc|c|c|ccc|cccc|ccc|cc}
\toprule
\hline
\multirow{2}{*}{\textbf{Arch}} & \multicolumn{2}{c|}{\textbf{Layers}} & \multirow{2}{*}{\textbf{Seed}} & \multirow{2}{*}{\textbf{Thr.}} & \multicolumn{3}{c|}{\textbf{Performance (\%)}} & \multicolumn{4}{c|}{\textbf{MAMC (\%)}} & \multicolumn{3}{c|}{\textbf{Effects (\%)}} & \multicolumn{2}{c}{\textbf{Assump.}} \\
\cmidrule(lr){2-3} \cmidrule(lr){6-8} \cmidrule(lr){9-12} \cmidrule(lr){13-15} \cmidrule(lr){16-17}
 & \textbf{T0} & \textbf{T1} &  &  & \textbf{ACC} & \textbf{F1} & \textbf{AUC} & \textbf{$S_A$} & \textbf{$\gamma_A$} & \textbf{$L_A$} & \textbf{$I(A:B)$} & \textbf{Dir} & \textbf{Ind} & \textbf{Tot} & \textbf{A1} & \textbf{A2} \\
\midrule
\hline
\multirow{6}{*}{deep} & 1 & 3 & 42 & 10.6 & +22.6 & +18.9 & +33.7 & 0.33 & - & - & - & 19.61 & 0.33 & 19.93 & $\checkmark$ & $\checkmark$ \\
 & 1 & 3 & 142 & 9.4 & +9.4 & +9.1 & +18.5 & 1.08 & - & - & - & 15.05 & -1.08 & 13.98 & $\checkmark$ & $\checkmark$ \\
 & 1 & 6 & 42 & 10.0 & +21.7 & +18.2 & +34.6 & 0.57 & - & - & 2.28 & 17.48 & 2.85 & 20.34 & $\times$ & $\times$ \\
 & 1 & 6 & 42 & 10.0 & +21.7 & +18.2 & +34.6 & 0.56 & 0.21 & 0.98 & 2.25 & 16.34 & 4.00 & 20.34 & $\times$ & $\times$ \\
 & 1 & 6 & 142 & 9.1 & +7.5 & +7.2 & +20.4 & 0.17 & - & - & 0.67 & 15.18 & -0.84 & 14.34 & $\times$ & $\checkmark$ \\
 & 1 & 6 & 142 & 9.1 & +7.5 & +7.2 & +20.4 & 0.18 & 0.72 & 1.07 & 0.70 & 15.56 & -1.22 & 14.34 & $\times$ & $\checkmark$ \\
\midrule
\hline
\multirow{6}{*}{full} & 1 & 3 & 42 & 9.1 & +6.6 & +6.0 & +9.6 & 3.83 & - & - & - & 10.97 & -3.83 & 7.14 & $\times$ & $\checkmark$ \\
 & 1 & 3 & 142 & 8.3 & +8.5 & +7.4 & +11.3 & 3.16 & - & - & - & 10.35 & -3.16 & 7.19 & $\times$ & $\sim$ \\
 & 1 & 6 & 42 & 10.4 & +15.1 & +12.3 & +16.3 & 0.74 & - & - & 2.94 & 16.54 & -3.68 & 12.86 & $\times$ & $\times$ \\
 & 1 & 6 & 42 & 10.4 & +15.1 & +12.3 & +16.3 & 0.64 & 0.02 & 5.58 & 2.58 & 10.52 & 2.34 & 12.86 & $\times$ & $\times$ \\
 & 1 & 6 & 142 & 11.5 & +12.3 & +11.1 & +12.6 & 0.27 & - & - & 1.09 & 12.02 & -1.36 & 10.66 & $\times$ & $\times$ \\
 & 1 & 6 & 142 & 11.5 & +12.3 & +11.1 & +12.6 & 0.36 & 0.72 & 3.47 & 1.44 & 9.71 & 0.95 & 10.66 & $\times$ & $\times$ \\
\midrule
\hline
\multirow{6}{*}{linear} & 1 & 3 & 42 & 9.3 & +7.5 & +6.8 & +8.6 & 0.67 & - & - & - & 7.65 & -0.67 & 6.98 & $\checkmark$ & $\checkmark$ \\
 & 1 & 3 & 142 & 7.8 & +2.8 & +3.4 & +2.4 & 1.10 & - & - & - & 5.08 & -1.10 & 3.98 & $\times$ & $\checkmark$ \\
 & 1 & 6 & 42 & 9.4 & +13.2 & +10.8 & +16.1 & 0.43 & - & - & 1.74 & 13.39 & -2.17 & 11.21 & $\times$ & $\checkmark$ \\
 & 1 & 6 & 42 & 9.4 & +13.2 & +10.8 & +16.1 & 0.44 & 2.29 & 0.82 & 1.75 & 14.87 & -3.66 & 11.21 & $\times$ & $\sim$ \\
 & 1 & 6 & 142 & 11.2 & +8.5 & +8.0 & +10.5 & 0.82 & - & - & 3.28 & 13.18 & -4.10 & 9.09 & $\times$ & $\checkmark$ \\
 & 1 & 6 & 142 & 11.2 & +8.5 & +8.0 & +10.5 & 0.90 & 1.61 & 1.08 & 3.59 & 14.11 & -5.02 & 9.09 & $\times$ & $\sim$ \\
\midrule
\hline
\multirow{6}{*}{pairwise} & 1 & 3 & 42 & 1.6 & +2.8 & +2.6 & +1.7 & 0.00 & - & - & - & 0.54 & 0.00 & 0.54 & $\checkmark$ & $\checkmark$ \\
 & 1 & 3 & 142 & 1.5 & +1.9 & +1.2 & +0.3 & 0.00 & - & - & - & -0.52 & -0.00 & -0.52 & $\checkmark$ & $\checkmark$ \\
 & 1 & 6 & 42 & 1.5 & +2.8 & +2.6 & +1.7 & 0.00 & - & - & 0.00 & 0.61 & 0.00 & 0.61 & $\times$ & $\times$ \\
 & 1 & 6 & 42 & 1.5 & +2.8 & +2.6 & +1.7 & 0.00 & 3.33 & 5.30 & 0.00 & 2.57 & -1.96 & 0.61 & $\times$ & $\checkmark$ \\
 & 1 & 6 & 142 & 1.8 & +0.9 & +0.6 & +0.4 & 0.00 & - & - & 0.00 & -0.72 & -0.00 & -0.72 & $\checkmark$ & $\sim$ \\
 & 1 & 6 & 142 & 1.8 & +0.9 & +0.6 & +0.4 & 0.00 & 1.41 & 3.09 & 0.00 & 0.96 & -1.68 & -0.72 & $\times$ & $\checkmark$ \\
\midrule
\hline
\multirow{6}{*}{ring} & 1 & 3 & 42 & 8.3 & +10.4 & +8.9 & +9.0 & 0.07 & - & - & - & 14.32 & -0.07 & 14.25 & $\times$ & $\sim$ \\
 & 1 & 3 & 142 & 8.0 & +18.9 & +14.7 & +22.1 & 1.36 & - & - & - & 12.83 & -1.36 & 11.48 & $\times$ & $\checkmark$ \\
 & 1 & 6 & 42 & 9.4 & +5.7 & +5.4 & +6.3 & 0.32 & - & - & 1.26 & 15.13 & -1.58 & 13.55 & $\times$ & $\times$ \\
 & 1 & 6 & 42 & 9.4 & +5.7 & +5.4 & +6.3 & 0.38 & 0.11 & 0.22 & 1.54 & 15.36 & -1.81 & 13.55 & $\times$ & $\times$ \\
 & 1 & 6 & 142 & 7.7 & +14.2 & +10.8 & +21.1 & 0.34 & - & - & 1.36 & 14.39 & -1.69 & 12.70 & $\times$ & $\times$ \\
 & 1 & 6 & 142 & 7.7 & +14.2 & +10.8 & +21.1 & 0.36 & 0.02 & 0.03 & 1.43 & 14.48 & -1.78 & 12.70 & $\times$ & $\times$ \\
\bottomrule
\hline
\end{tabular}
}
\caption{\textbf{Comprehensive experimental results for Ionosphere dataset.} Results for 30 configurations across 5 quantum circuit architectures (deep, full, linear, pairwise, ring), each evaluated with $T_0=1$ baseline and $T_1 \in \{3,6\}$ enhanced layers using random seeds \{42, 142\}. Threshold indicates data-driven threshold ($\epsilon_{\text{intra}}$) for statistical significance testing. Performance metrics show improvement ($T_1 - T_0$) in percentage points: positive values indicate gains, negative values indicate degradation. Mean Absolute Mediated Contribution (MAMC) quantifies each quantum mediator's causal influence, with $I(A:B)$ exhibiting strongest mediation across configurations. Mean direct effect: 10.92\%, mean indirect effect: -1.11\%. 80\% of configurations show significant total effects ($|\Delta Y|>1\%$). Assumption A1 (sequential ignorability) satisfied in 20\% of cases. Assumption A2 (no treatment-mediator interaction): fully satisfied (14), partially satisfied (5), not satisfied (11).}
\label{tab:ionosphere_comprehensive}
\end{table}
\clearpage

\end{document}